\documentclass[a4paper,USenglish]{lipics-v2018}

\usepackage{microtype}

\nolinenumbers

\hideLIPIcs

\usepackage{amsmath, mathtools}
\usepackage{nicefrac}

\usepackage[usenames,dvipsnames,svgnames,table]{xcolor}
\definecolor{jazzberryjam}{rgb}{0.65, 0.04, 0.37}
\definecolor{egyptianblue}{rgb}{0.06, 0.2, 0.65}

\usepackage{tikz}
\usetikzlibrary{patterns}
\usetikzlibrary{decorations.pathreplacing}
\usetikzlibrary{math}

\theoremstyle{plain}
\newtheorem{algorithm}{Algorithm}
\newtheorem{observation}{Observation}

\newcommand{\Cmax}{C_{\max}}

\newcommand{\jobs}{\mathcal{J}}
\newcommand{\confs}{\mathcal{C}}
\newcommand{\size}{\Lambda}
\newcommand{\mods}{\mathcal{M}}
\newcommand{\bobs}{\mathcal{B}}
\newcommand{\groups}{\Gamma}
\newcommand{\msizes}{H}
\newcommand{\layers}{\Xi}
\newcommand{\MCIP}{\mathrm{MCIP}}
\newcommand{\bsjobs}{\jobs^{\mathrm{bst}}}
\newcommand{\msjobs}{\jobs^{\mathrm{mst}}}
\newcommand{\ssjobs}{\jobs^{\mathrm{sst}}}
\newcommand{\bsmods}{\mods^{\mathrm{bst}}}
\newcommand{\msmods}{\mods^{\mathrm{mst}}}
\newcommand{\ssmods}{\mods^{\mathrm{sst}}}
\newcommand{\ptiny}{\mathrm{tiny}}
\newcommand{\pbig}{\mathrm{big}}
\newcommand{\plarge}{\mathrm{large}}
\newcommand{\psmall}{\mathrm{small}}

\newcommand{\jparts}{\jobs'}
\newcommand{\jfrac}{\lambda}
\newcommand{\jcard}{\kappa}

\newcommand{\poly}{\mathrm{poly}}

\newcommand{\Opt}{\mathrm{OPT}}
\newcommand{\Oh}{\mathcal{O}}
\newcommand{\eps}{\varepsilon}

\newcommand{\ZZ}{\mathbb{Z}}
\newcommand{\QQ}{\mathbb{Q}}

\DeclarePairedDelimiter\floor{\lfloor}{\rfloor}
\DeclarePairedDelimiter\ceil{\lceil}{\rceil}
\DeclarePairedDelimiter\set{\lbrace}{\rbrace}
\DeclarePairedDelimiterX\sett[2]{\lbrace}{\rbrace}{ #1 \,\delimsize| \,\mathopen{} #2 }

\title{Empowering the Configuration-IP -- New PTAS Results for Scheduling with Setups Times}

\titlerunning{Empowering the Configuration-IP}

\author{Klaus Jansen}{Department of Computer Science, Kiel University, Kiel, Germany}{kj@informatik.uni-kiel.de}{}{German Research Foundation (DFG) project JA 612/20-1}

\author{Kim-Manuel Klein}{Department of Computer Science, Kiel University, Kiel, Germany}{kmk@informatik.uni-kiel.de}{}{}

\author{Marten Maack}{Department of Computer Science, Kiel University, Kiel, Germany}{mmaa@informatik.uni-kiel.de}{}{}

\author{Malin Rau}{Department of Computer Science, Kiel University, Kiel, Germany}{mra@informatik.uni-kiel.de}{}{}

\authorrunning{K. Jansen, K. Klein, M. Maack and M. Rau}

\Copyright{Klaus Jansen, Kim-Manuel Klein, Marten Maack and Malin Rau}
\subjclass{, }


\acknowledgements{We thank Syamantak Das for helpful discussions on the problem.}

\EventEditors{John Q. Open and Joan R. Access}
\EventNoEds{2}
\EventLongTitle{42nd Conference on Very Important Topics (CVIT 2016)}
\EventShortTitle{CVIT 2016}
\EventAcronym{CVIT}
\EventYear{2016}
\EventDate{December 24--27, 2016}
\EventLocation{Little Whinging, United Kingdom}
\EventLogo{}
\SeriesVolume{42}
\ArticleNo{23}

\begin{document}


\maketitle

\begin{abstract}
Integer linear programs of configurations, or configuration IPs, are a classical tool in the design of algorithms for scheduling and packing problems, where a set of items has to be placed in multiple target locations. 
Herein a configuration describes a possible placement on one of the target locations, and the IP is used to chose suitable configurations covering the items. 
We give an augmented IP formulation, which we call the module configuration IP. 
It can be described within the framework of n-fold integer programming and therefore be solved efficiently. 
As an application, we consider scheduling problems with setup times, in which a set of jobs has to be scheduled on a set of identical machines, with the objective of minimizing the makespan. 
For instance, we investigate the case that jobs can be split and scheduled on multiple machines. 
However, before a part of a job can be processed an uninterrupted setup depending on the job has to be paid. 
For both of the variants that jobs can be executed in parallel or not, we obtain an efficient polynomial time approximation scheme (EPTAS) of running time $f(1/\eps)\times \poly(|I|)$ with a single exponential term in $f$ for the first and a double exponential one for the second case. 
Previously, only constant factor approximations of $5/3$ and $4/3 + \eps$ respectively were known.
Furthermore, we present an EPTAS for a problem where classes of (non-splittable) jobs are given, and a setup has to be paid for each class of jobs being executed on one machine. 
\end{abstract}

\section{Introduction}

In this paper, we present an augmented formulation of the classical integer linear program of configurations (configuration IP) and demonstrate its use in the design of efficient polynomial time approximation schemes for scheduling problems with setup times.
Configuration IPs are widely used in the context of scheduling or packing problems, in which items have to be distributed to multiple target locations.
The configurations describe possible placements on a single location, and the integer linear program (IP) is used to choose a proper selection covering all items.
Two fundamental problems, for which configuration IPs have prominently been used, are bin packing and minimum makespan scheduling on identical parallel machines, or machine scheduling for short.
For bin packing, the configuration IP was introduced as early as 1961 by Gilmore and Gomory \cite{gilmore1961linear}, and the recent results for both problems typically use configuration IPs as a core technique, see, e.g., \cite{goemans2014polynomiality,JKV16ICALP}.
In the present work, we consider scheduling problems and therefore introduce the configuration IP in more detail using the example of machine scheduling.

\subparagraph*{Configuration IP for Machine Scheduling.}

In the problem of machine scheduling, a set $\jobs$ of $n$ jobs is given together with processing times $p_j$ for each job $j$ and a number $m$ of identical machines.
The objective is to find a schedule $\sigma: \jobs \rightarrow [m]$, such that the makespan is minimized, that is, the latest finishing time of any job $C_{\max}(\sigma)=\max_{i\in[m]}\sum_{j\in\sigma^{-1}(i)}p_j$.
For a given makespan bound, the configurations may be defined as multiplicity vectors indexed by the occurring processing times, where the overall length of the chosen processing times does not violate the bound.
The configuration IP is then given by variables $x_C$ for each configuration $C$; constraints ensuring that there is a machine for each configuration, i.e., $\sum_{C}x_C = m$; and further constraints due to which the jobs are covered, i.e., $\sum_{C}C_p x_C = |\sett{j\in\jobs}{p_j=p}|$ for each processing time $p$.
In combination with certain simplification techniques, this type of IP is often used in the design of \emph{polynomial time approximation schemes} (PTAS).
A PTAS is a procedure that, for any fixed accuracy parameter $\eps > 0$, returns a solution with approximation guarantee $(1+\eps)$ that is, a solution, whose objective value lies within a factor of $(1 + \eps)$ of the optimum.
In the context of machine scheduling, the aforementioned simplification techniques can be used to guess the target makespan $T$ of the given instance; to upper bound the cardinality of the set of processing times $P$ by a constant (depending in $1/\eps$); and to lower bound the processing times in size, such that they are within a constant factor of the makespan $T$ (see, e.g., \cite{alon1998approximation,JKV16ICALP}).
Hence, only a constant number of configurations is needed, yielding an integer program with a constant number of variables.
Integer programs of that kind can be efficiently solved using the classical algorithm by Lenstra and Kannan \cite{lenstra1983integer,kannan1987minkowski}, yielding a PTAS for machine scheduling.
Here, the error of $(1+\eps)$ in the quality of the solution is due to the simplification steps, and the scheme has a running time of the form $f(1/\eps)\times\poly(|I|)$, where $|I|$ denotes the input size, and $f$ some computable function.
A PTAS with this property is called \emph{efficient} (EPTAS).
Note that for a regular PTAS a running time of the form $|I|^{f(1/\epsilon)}$ is allowed.
It is well known, that machine scheduling is strongly NP-hard, and therefore there is no optimal polynomial time algorithm, unless P$=$NP, and also a so-called \emph{fully polynomial} PTAS (FPTAS)---which is an EPTAS with a polynomial function~$f$---cannot be hoped for.

\subparagraph*{Machine Scheduling with Classes.}

The configuration IP is used in a wide variety of approximation schemes for machine scheduling problems \cite{alon1998approximation, JKV16ICALP}. 
However, the approach often ceases to work for scheduling problems in which the jobs have to fulfill some additional requirements, like, for instance, class dependencies.
A problem emerging, in this case, is that the additional requirements have to be represented in the configurations, resulting in a super-constant number of variables in the IP.
We elaborate on this using a concrete example:
Consider the variant of machine scheduling in which the jobs are partitioned into $K$ setup classes.
For each job $j$ a class $k_j$ is given and for each class $k$ a setup time $s_k$ has to be paid on a machine, if a job belonging to that class is scheduled on it, i.e., $\Cmax(\sigma) = \max_{i\in[m]}\big( \sum_{j\in\sigma^{-1}(i)} p_j + \sum_{k\in\sett{k_j}{j\in\sigma{-1}(i)}} s_k \big)$.
With some effort, simplification steps similar to the ones for machine scheduling can be applied.
In the course of this, the setup times as well can be bounded in number and guaranteed to be sufficiently big \cite{SetupPTAS2016}.
However, it is not hard to see that the configuration IP still cannot be trivially extended, while preserving its solvability.
For instance, extending the configurations with multiplicities of setup times will not work, because then we have to make sure that a configuration is used for a fitting subset of classes, creating the need to encode class information into the configurations or introduce other class dependent variables.

\subparagraph*{Module Configuration IP.}

Our approach to deal with the class dependencies of the jobs is to cover the job classes with so-called modules and cover the modules in turn with configurations in an augmented IP called the module configuration IP ($\MCIP$).
In the setup class model, for instance, the modules may be defined as combinations of setup times and configurations of processing times, and the actual configurations as multiplicity vectors of module sizes.
The number of both the modules and the configurations will typically be bounded by a constant.
To cover the classes by modules each class is provided with its own set of modules, that is, there are variables for each pair of class and module. 
Since the number of classes is part of the input, the number of variables in the resulting $\MCIP$ is super-constant, and therefore the algorithm by Lenstra and Kannan \cite{lenstra1983integer,kannan1987minkowski} is not the proper tool for the solving of the $\MCIP$.
However, the $\MCIP$ has a certain simple structure:
The mentioned variables are partitioned into uniform classes each corresponding to the set of modules, and for each class, the modules have to do essentially the same---cover the jobs of the class.
Utilizing these properties, we can formulate the $\MCIP$ in the framework of $n$-fold integer programms---a class of IPs whose variables and constraints fulfill certain uniformity requirements.
In 2013 Hemmecke, Onn, and Romanchuk \cite{nfoldcubic} showed that $n$-fold IPs can be efficiently solved, and very recently both Eisenbrand, Hunkenschröder and Klein \cite{Kim-n-fold2018} and independently Koutecký, Levin and Onn \cite{KLS-n-fold-2018} developed algorithms with greatly improved running times for the problem.
For a detailed description of the $\MCIP$, the reader is referred to Section \ref{sec:ILP}.
In Figure \ref{fig:MCIP} the basic idea of the $\MCIP$ is visualized. 

Using the $\MCIP$, we are able to formulate an EPTAS for machine scheduling in the setup class model described above.
Before, only a regular PTAS with running time $nm^{\Oh(1/\eps^5)}$ was known \cite{SetupPTAS2016}.
To the best of our knowledge, this is the first use of $n$-fold integer programing in the context of approximation algorithms. 

\begin{figure}
\resizebox{\textwidth}{!}{%
\begin{tikzpicture}
\pgfmathsetmacro{\w}{0.25}
\pgfmathsetmacro{\h}{0.25}
\pgfmathsetmacro{\dist}{1.5}
\foreach \x/\y/\z/\k in {0/Green/0.23/0.2,1/Purple/0.21/0.1,2/TealBlue/0.14/0.1,3/Cerulean/0.15/0.3,4/Dandelion/0.24/0.2,5/Maroon/0.03/0.1}{
    \draw[fill = \y, opacity = 0.6] (\dist*\x*\w,0) rectangle (\dist*\x*\w +\w,\x*\h +\h);
    \draw[] (\dist*\x*\w,0) rectangle (\dist*\x*\w +\w,\x*\h +\h);
    	\node at (\dist*\x*\w+ 0.5*\w,-2*\h) {\small{$p_{\x}$}};
    	\begin{scope}[yshift = 7*\h cm]
    	\draw (\dist*\x*\w,11*\h) rectangle (\dist*\x*\w +\w,20*\h);
    	\node at (\dist*\x*\w+ 0.5*\w,9*\h) {\small{$x_{\x}$}};
    	\end{scope}
    	}
\begin{scope}[yshift = 7*\h cm]
\draw[fill = Maroon, opacity = 0.6] (\dist*0*\w,11*\h) rectangle (\dist*0*\w +\w,11*\h + 6*\h);
\draw[fill = TealBlue, opacity = 0.6] (\dist*0*\w,11*\h+ 6*\h) rectangle (\dist*0*\w +\w,11*\h + 9*\h);  
\draw[fill = Dandelion, opacity = 0.6] (\dist*1*\w,11*\h) rectangle (\dist*1*\w +\w,11*\h + 5*\h);	
\draw[fill = TealBlue, opacity = 0.6] (\dist*1*\w,11*\h + 5*\h) rectangle (\dist*1*\w +\w,11*\h + 8*\h);
\draw[fill = Green, opacity = 0.6] (\dist*1*\w,11*\h + 8*\h) rectangle (\dist*1*\w +\w,11*\h + 9*\h);		
\draw[fill = Maroon, opacity = 0.6] (\dist*2*\w,11*\h) rectangle (\dist*2*\w +\w,11*\h + 6*\h);	
\draw[fill = Purple, opacity = 0.6] (\dist*2*\w,11*\h + 6*\h) rectangle (\dist*2*\w +\w,11*\h + 8*\h);		
\draw[fill = Green, opacity = 0.6] (\dist*2*\w,11*\h + 8*\h) rectangle (\dist*2*\w +\w,11*\h + 9*\h);	    	
\draw[fill = Cerulean, opacity = 0.6] (\dist*3*\w,11*\h) rectangle (\dist*3*\w +\w,11*\h + 4*\h);	
\draw[fill = TealBlue, opacity = 0.6] (\dist*3*\w,11*\h + 4*\h) rectangle (\dist*3*\w +\w,11*\h + 7*\h);	
\draw[fill = Purple, opacity = 0.6] (\dist*3*\w,11*\h+ 7*\h) rectangle (\dist*3*\w +\w,11*\h + 9*\h);
\draw[fill = Dandelion, opacity = 0.6] (\dist*4*\w,11*\h) rectangle (\dist*4*\w +\w,11*\h + 5*\h);	
\draw[fill = Cerulean, opacity = 0.6] (\dist*4*\w,11*\h + 5*\h) rectangle (\dist*4*\w +\w,11*\h + 9*\h);	
\draw[fill = TealBlue, opacity = 0.6] (\dist*5*\w,11*\h) rectangle (\dist*5*\w +\w,11*\h + 3*\h);	
\draw[fill = TealBlue, opacity = 0.6] (\dist*5*\w,11*\h + 3*\h) rectangle (\dist*5*\w +\w,11*\h + 6*\h);
\draw[fill = TealBlue, opacity = 0.6] (\dist*5*\w,11*\h + 6*\h) rectangle (\dist*5*\w +\w,11*\h + 9*\h);
\end{scope}
\node[text width=2cm] at (-7*\w,-2*\h) {s.c. mult.};
\node[text width=2cm] at (-7*\w,1.5*\h) {constant\\ sizes};
\node[text width=2cm] at (-7*\w,19*\h) {constant\\ configurations};
\node[text width=2cm] at (-7*\w,16*\h) {s.c. mult.};


\begin{scope}[xshift = 20*\w cm]
\pgfmathsetmacro{\w}{0.25}
\pgfmathsetmacro{\h}{0.25}
\pgfmathsetmacro{\dist}{1.6}
\pgfmathsetmacro{\distClass}{7}
\foreach \l/\p in {0/crosshatch,1/crosshatch dots,2/north east lines}{
	\foreach \x/\y/\z in {0/Green/-1,1/Purple/1,2/TealBlue/-1,3/Cerulean/1,4/Dandelion/-1,5/Maroon/1}{
    		\draw[pattern = \p, pattern color = black] (\dist*\x*\w+\dist*\l*\distClass*\w,0) rectangle (\dist*\x*\w +\w+\dist*\l*\distClass*\w,\x*\h +\h);
    			\node at (\dist*\x*\w +0.5*\w+\dist*\l*\distClass*\w,-1*\h-\h) {\small{$p_{\x}^{\l}$}};
    			
    }
}

\foreach \x/\y/\z/\k in {0, ..., 5}{
    	\begin{scope}[yshift = 7*\h cm]
    	\node at (\dist*\x*\w +0.5*\w+\dist*1*\distClass*\w,9*\h) {\small{$x_{\x}$}};
    	\end{scope}
    	}

\node at (\dist*2*\distClass*\w + 0.9*\dist*\distClass*\w,0) {$\dots$};
\node[text width=2cm] at (-4*\w,2.5*\h) {s.c. classes\\ with \\ c. sizes};    
\node[text width=2cm] at (-4*\w,-2*\h) {s.c. mult.};	
\node[text width=2cm] at (-4*\w,13*\h) {constant \\ module \\ sizes/c.};
\node[text width=2cm] at (-4*\w,8*\h) {modules \\ per class};  
\node[text width=2cm] at (-4*\w,19*\h) {constant \\ configurations};

\draw[fill = Green, opacity = 0.6] (\dist*0*\w+\dist*0*\distClass*\w,10*\h) rectangle (\dist*0*\w +\w+\dist*0*\distClass*\w,10*\h +\h);
\draw[fill = Green, opacity = 0.6] (\dist*1*\w+\dist*0*\distClass*\w,10*\h) rectangle (\dist*1*\w +\w+\dist*0*\distClass*\w,10*\h +\h);
\draw[fill = Green, opacity = 0.6] (\dist*2*\w+\dist*0*\distClass*\w,10*\h) rectangle (\dist*2*\w +\w+\dist*0*\distClass*\w,10*\h +\h);
\draw[pattern = crosshatch] (\dist*0*\w+\dist*0*\distClass*\w,10*\h) rectangle (\dist*0*\w +\w+\dist*0*\distClass*\w,10*\h +\h);
\draw[pattern = crosshatch dots] (\dist*1*\w+\dist*0*\distClass*\w,10*\h) rectangle (\dist*1*\w +\w+\dist*0*\distClass*\w,10*\h +\h);
\draw[pattern = north east lines] (\dist*2*\w+\dist*0*\distClass*\w,10*\h) rectangle (\dist*2*\w +\w+\dist*0*\distClass*\w,10*\h +\h);

\draw[fill = Purple, opacity = 0.6] (\dist*3*\w+\dist*0*\distClass*\w,10*\h) rectangle (\dist*3*\w +\w+\dist*0*\distClass*\w,10*\h +\h);
\draw[fill = Purple, opacity = 0.6] (\dist*3*\w+\dist*0*\distClass*\w,11*\h) rectangle (\dist*3*\w +\w+\dist*0*\distClass*\w,11*\h +\h);
\draw[pattern = crosshatch] (\dist*3*\w+\dist*0*\distClass*\w,10*\h) rectangle (\dist*3*\w +\w+\dist*0*\distClass*\w,10*\h +\h);
\draw[pattern = crosshatch] (\dist*3*\w+\dist*0*\distClass*\w,11*\h) rectangle (\dist*3*\w +\w+\dist*0*\distClass*\w,11*\h +\h);

\draw[fill = Purple, opacity = 0.6] (\dist*4*\w+\dist*0*\distClass*\w,10*\h) rectangle (\dist*4*\w +\w+\dist*0*\distClass*\w,10*\h +\h);
\draw[fill = Purple, opacity = 0.6] (\dist*4*\w+\dist*0*\distClass*\w,11*\h) rectangle (\dist*4*\w +\w+\dist*0*\distClass*\w,11*\h +\h);
\draw[pattern = crosshatch dots] (\dist*4*\w+\dist*0*\distClass*\w,10*\h) rectangle (\dist*4*\w +\w+\dist*0*\distClass*\w,10*\h +\h);
\draw[pattern = crosshatch dots] (\dist*4*\w+\dist*0*\distClass*\w,11*\h) rectangle (\dist*4*\w +\w+\dist*0*\distClass*\w,11*\h +\h);

\draw[fill = Purple, opacity = 0.6] (\dist*5*\w+\dist*0*\distClass*\w,10*\h) rectangle (\dist*5*\w +\w+\dist*0*\distClass*\w,10*\h +\h);
\draw[fill = Purple, opacity = 0.6] (\dist*5*\w+\dist*0*\distClass*\w,11*\h) rectangle (\dist*5*\w +\w+\dist*0*\distClass*\w,11*\h +\h);
\draw[pattern = north east lines] (\dist*5*\w+\dist*0*\distClass*\w,10*\h) rectangle (\dist*5*\w +\w+\dist*0*\distClass*\w,10*\h +\h);
\draw[pattern = north east lines] (\dist*5*\w+\dist*0*\distClass*\w,11*\h) rectangle (\dist*5*\w +\w+\dist*0*\distClass*\w,11*\h +\h);

\draw[fill = Purple, opacity = 0.6] (\dist*6*\w+\dist*0*\distClass*\w,10*\h) rectangle (\dist*6*\w +\w+\dist*0*\distClass*\w,11*\h +\h);
\draw[pattern = crosshatch] (\dist*6*\w+\dist*0*\distClass*\w,10*\h) rectangle (\dist*6*\w +\w+\dist*0*\distClass*\w,11*\h +\h);

\draw[fill = Purple, opacity = 0.6] (\dist*7*\w+\dist*0*\distClass*\w,10*\h) rectangle (\dist*7*\w +\w+\dist*0*\distClass*\w,11*\h +\h);
\draw[pattern = crosshatch dots] (\dist*7*\w+\dist*0*\distClass*\w,10*\h) rectangle (\dist*7*\w +\w+\dist*0*\distClass*\w,11*\h +\h);

\draw[fill = Purple, opacity = 0.6] (\dist*8*\w+\dist*0*\distClass*\w,10*\h) rectangle (\dist*8*\w +\w+\dist*0*\distClass*\w,11*\h +\h);
\draw[pattern = north east lines] (\dist*8*\w+\dist*0*\distClass*\w,10*\h) rectangle (\dist*8*\w +\w+\dist*0*\distClass*\w,11*\h +\h);

\draw[fill = TealBlue , opacity = 0.6] (\dist*9*\w+\dist*0*\distClass*\w,10*\h) rectangle (\dist*9*\w +\w+\dist*0*\distClass*\w,11*\h +\h);
\draw[fill = TealBlue, opacity = 0.6 ] (\dist*9*\w+\dist*0*\distClass*\w,12*\h) rectangle (\dist*9*\w +\w+\dist*0*\distClass*\w,12*\h +\h);
\draw[pattern = crosshatch] (\dist*9*\w+\dist*0*\distClass*\w,10*\h) rectangle (\dist*9*\w +\w+\dist*0*\distClass*\w,11*\h +\h);
\draw[pattern = crosshatch] (\dist*9*\w+\dist*0*\distClass*\w,12*\h) rectangle (\dist*9*\w +\w+\dist*0*\distClass*\w,12*\h +\h);

\draw[fill = TealBlue, opacity = 0.6] (\dist*10*\w+\dist*0*\distClass*\w,10*\h) rectangle (\dist*10*\w +\w+\dist*0*\distClass*\w,12*\h +\h);
\draw[pattern = crosshatch dots] (\dist*10*\w+\dist*0*\distClass*\w,10*\h) rectangle (\dist*10*\w +\w+\dist*0*\distClass*\w,12*\h +\h);

\draw[fill = TealBlue, opacity = 0.6] (\dist*11*\w+\dist*0*\distClass*\w,10*\h) rectangle (\dist*11*\w +\w+\dist*0*\distClass*\w,12*\h +\h);
\draw[pattern = north east lines] (\dist*11*\w+\dist*0*\distClass*\w,10*\h) rectangle (\dist*11*\w +\w+\dist*0*\distClass*\w,10*\h +\h);
\draw[pattern = north east lines] (\dist*11*\w+\dist*0*\distClass*\w,10*\h) rectangle (\dist*11*\w +\w+\dist*0*\distClass*\w,10*\h +2*\h);
\draw[pattern = north east lines] (\dist*11*\w+\dist*0*\distClass*\w,10*\h) rectangle (\dist*11*\w +\w+\dist*0*\distClass*\w,10*\h +3*\h);

\node at (\dist*13*\w+\dist*0*\distClass*\w,10*\h) {$\dots$};

\draw[fill = Cerulean, opacity = 0.6] (\dist*14*\w+\dist*0*\distClass*\w,10*\h) rectangle (\dist*14*\w +\w+\dist*0*\distClass*\w,13*\h +\h);
\draw[fill = Cerulean, opacity = 0.6] (\dist*15*\w+\dist*0*\distClass*\w,10*\h) rectangle (\dist*15*\w +\w+\dist*0*\distClass*\w,10*\h +2*\h);
\draw[fill = Cerulean, opacity = 0.6] (\dist*15*\w+\dist*0*\distClass*\w,10*\h +2*\h) rectangle (\dist*15*\w +\w+\dist*0*\distClass*\w,13*\h +\h);
\draw[pattern = crosshatch] (\dist*14*\w+\dist*0*\distClass*\w,10*\h) rectangle (\dist*14*\w +\w+\dist*0*\distClass*\w,10*\h +\h);
\draw[pattern = crosshatch] (\dist*14*\w+\dist*0*\distClass*\w,10*\h+\h) rectangle (\dist*14*\w +\w+\dist*0*\distClass*\w,10*\h +2*\h);
\draw[pattern = crosshatch] (\dist*14*\w+\dist*0*\distClass*\w,10*\h+2*\h) rectangle (\dist*14*\w +\w+\dist*0*\distClass*\w,13*\h +\h);
\draw[pattern = crosshatch dots] (\dist*15*\w+\dist*0*\distClass*\w,10*\h) rectangle (\dist*15*\w +\w+\dist*0*\distClass*\w,13*\h +\h);

\node at (\dist*17*\w+\dist*0*\distClass*\w,10*\h) {$\dots$};

\draw[fill = Dandelion, opacity = 0.6] (\dist*18*\w+\dist*0*\distClass*\w,10*\h) rectangle (\dist*18*\w +\w+\dist*0*\distClass*\w,14*\h +\h);
\draw[pattern = crosshatch dots] (\dist*18*\w+\dist*0*\distClass*\w,10*\h) rectangle (\dist*18*\w +\w+\dist*0*\distClass*\w,10*\h +3*\h);
\draw[pattern = crosshatch dots] (\dist*18*\w+\dist*0*\distClass*\w,10*\h+3*\h) rectangle (\dist*18*\w +\w+\dist*0*\distClass*\w,14*\h +\h);

\node at (\dist*20*\w+\dist*0*\distClass*\w,10*\h) {$\dots$};

\draw[fill = Maroon, opacity = 0.6] (\dist*21*\w+\dist*0*\distClass*\w,10*\h) rectangle (\dist*21*\w +\w+\dist*0*\distClass*\w,15*\h +\h);
\draw[pattern = north east lines] (\dist*21*\w+\dist*0*\distClass*\w,10*\h) rectangle (\dist*21*\w +\w+\dist*0*\distClass*\w,15*\h +\h);

\node at (\dist*23*\w+\dist*0*\distClass*\w,10*\h) {$\dots$};

\draw[fill = Maroon, opacity = 0.6] (\dist*0*\w+\dist*1*\distClass*\w,18*\h) rectangle (\dist*0*\w+\dist*1*\distClass*\w + \w,18*\h +6*\h);
\draw[fill = TealBlue, opacity = 0.6] (\dist*0*\w+\dist*1*\distClass*\w,18*\h +6*\h) rectangle (\dist*0*\w+\dist*1*\distClass*\w + \w,18*\h +3*\h +6*\h);

\draw[fill = Dandelion, opacity = 0.6] (\dist*1*\w+\dist*1*\distClass*\w,18*\h) rectangle (\dist*1*\w+\dist*1*\distClass*\w+\w,18*\h +5*\h);
\draw[fill = TealBlue, opacity = 0.6] (\dist*1*\w+\dist*1*\distClass*\w,18*\h +5*\h) rectangle (\dist*1*\w+\dist*1*\distClass*\w+\w,18*\h +3*\h +5*\h);
\draw[fill = Green, opacity = 0.6] (\dist*1*\w+\dist*1*\distClass*\w,18*\h +5*\h+3*\h) rectangle (\dist*1*\w+\dist*1*\distClass*\w+\w,18*\h +3*\h +5*\h +1*\h);

\draw[fill = Maroon, opacity = 0.6] (\dist*2*\w+\dist*1*\distClass*\w,18*\h) rectangle (\dist*2*\w+\dist*1*\distClass*\w+\w,18*\h +6*\h);
\draw[fill = Purple, opacity = 0.6] (\dist*2*\w+\dist*1*\distClass*\w,18*\h +6*\h) rectangle (\dist*2*\w+\dist*1*\distClass*\w+\w,18*\h +2*\h +6*\h);
\draw[fill = Green, opacity = 0.6] (\dist*2*\w+\dist*1*\distClass*\w,18*\h +5*\h+3*\h) rectangle (\dist*2*\w+\dist*1*\distClass*\w+\w,18*\h +2*\h +6*\h +1*\h);

\draw[fill = Cerulean, opacity = 0.6] (\dist*3*\w+\dist*1*\distClass*\w,18*\h) rectangle (\dist*3*\w+\dist*1*\distClass*\w+\w,18*\h +4*\h);
\draw[fill = TealBlue, opacity = 0.6] (\dist*3*\w+\dist*1*\distClass*\w,18*\h +4*\h) rectangle (\dist*3*\w+\dist*1*\distClass*\w+\w,18*\h +3*\h +4*\h);
\draw[fill = Purple, opacity = 0.6] (\dist*3*\w+\dist*1*\distClass*\w,18*\h +4*\h+3*\h) rectangle (\dist*3*\w+\dist*1*\distClass*\w+\w,18*\h +3*\h +4*\h +2*\h);

\draw[fill = Dandelion, opacity = 0.6] (\dist*4*\w+\dist*1*\distClass*\w,18*\h) rectangle (\dist*4*\w+\dist*1*\distClass*\w+\w,18*\h +5*\h);
\draw[fill = Cerulean, opacity = 0.6] (\dist*4*\w+\dist*1*\distClass*\w,18*\h +5*\h) rectangle (\dist*4*\w+\dist*1*\distClass*\w+\w,18*\h +5*\h +4*\h);

\draw[fill = TealBlue, opacity = 0.6] (\dist*5*\w+\dist*1*\distClass*\w,18*\h) rectangle (\dist*5*\w+\dist*1*\distClass*\w+\w,18*\h +3*\h);
\draw[fill = TealBlue, opacity = 0.6] (\dist*5*\w+\dist*1*\distClass*\w,18*\h +3*\h) rectangle (\dist*5*\w+\dist*1*\distClass*\w+\w,18*\h +3*\h +3*\h);
\draw[fill = TealBlue, opacity = 0.6] (\dist*5*\w+\dist*1*\distClass*\w,18*\h +3*\h+3*\h) rectangle (\dist*5*\w+\dist*1*\distClass*\w+\w,18*\h +3*\h +3*\h +3*\h);

\end{scope}

\end{tikzpicture}
}
\caption{On the left, there is a schematic representation of the configuration IP.
There are constant different sizes each occurring a super-constant number of times.
The sizes are directly mapped to configurations. 
On the right, there is a schematic representation of the MCIP. 
There is a super-constant number of classes, each containing a constant number of sizes which have super-constant multiplicities. 
The elements from the class are mapped to a constant number of different modules, which have a constant number of sizes. 
These module sizes are mapped to configurations.}
\label{fig:MCIP}
\end{figure}
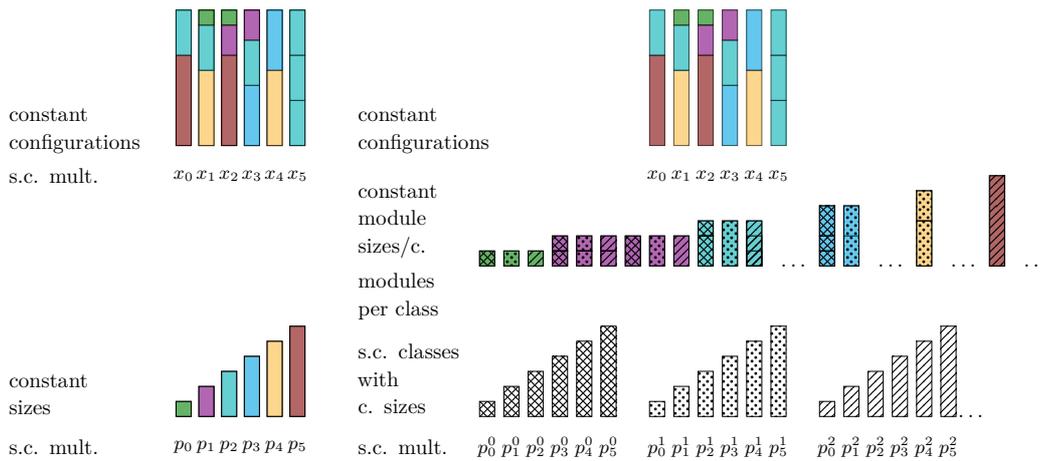

\subparagraph*{Results and Methodology.}

To show the conceptual power of the $\MCIP$, we utilize it for two more problems:
The \emph{splittable} and the \emph{preemptive} setup model of machine scheduling.
In both variants for each job $j$, a setup time $s_j$ is given.
Each job may be partitioned into multiple parts that can be assigned to different machines, but before any part of the job can be processed the setup time has to be paid.
In the splittable model, job parts belonging to the same job can be processed in parallel, and therefore beside the partition of the jobs, it suffices to find an assignment of the job parts to machines.
This is not the case for the preemptive model, in which additionally a starting time for each job part has to be found, and two parts of the same job may not be processed in parallel.
In 1999 Schuurman and Woeginger \cite{schuurman1999preemptive} presented a polynomial time algorithm for the preemptive model with approximation guarantee $4/3 + \eps$, and for the splittable case a guarantee of $5/3$ was achieved by Chen, Ye and Zhang \cite{chen2006lot}.
These are the best known approximation guarantees for the problems at hand.  
We show that solutions arbitrarily close to the optimum can be found in polynomial time:
\begin{theorem}\label{thm:main_result}
There is an efficient PTAS with running time $2^{f(1/\eps)}\poly(|I|)$ for minimum makespan scheduling on identical parallel machines in the setup-class model, as well as in the preemptive and splittable setup models. 
\end{theorem}
More precisely, we get a running time of $2^{\Oh(\nicefrac{1}{\eps^{3}}\log^4\nicefrac{1}{\eps})}K^2 n m \log (Km)$ in the setup class model, $2^{\Oh(\nicefrac{1}{\eps^2}\log^3\nicefrac{1}{\eps})}n^2\log^3 (nm)$ in the splittable, and $2^{2^{\Oh(\nicefrac{1}{\eps}\log\nicefrac{1}{\eps})}}n^2 m\log m\log (nm)$ in the preemptive model. 
Note, that all three problems are strongly NP-hard, due to trivial reductions from machine scheduling, and our results are therefore in some sense best possible.

Summing up, the main achievement of this work is the development of the module configuration IP and its application in the development of approximation schemes.
Up to now, EPTAS or even PTAS results seemed out of reach for the considered problems, and for the preemptive model we provide the first improvement in 20 years.
The simplification techniques developed for the splittable and preemptive model in order to employ the $\MCIP$ are original and in the latter case quite elaborate, and therefore interesting by themselfs.
Furthermore, we expect the $\MCIP$ to be applicable to other packing and scheduling problems as well, in particular for variants of machine scheduling and bin packing with additional class depended constraints.
On a more conceptual level, we gave a first demonstration of the potential of $n$-fold integer programming in the theory of approximation algorithms, and hope to inspire further studies in this direction. 

We conclude this paragraph with a more detailed overview of our results and their presentation.
For all three EPTAS results we employ the classical dual approximation framework by Hochbaum and Shmoys \cite{dualapprox} to get a guess of the makespan $T$.
This approach is introduced in Section \ref{sec:prelim} together with $n$-fold IPs and formal definitions of the problems.
In the following section, we develop the module configuration IP, in its basic form and argue that it is indeed an $n$-fold IP.
The EPTAS results follow the same basic approach described above for machine scheduling:
We find a schedule for a simplified instance via the $\MCIP$ and transform it into a schedule for the original one.
The simplification steps typically include rounding of the processing and setup times using standard techniques, as well as, the removal of certain jobs, which later can be reinserted via carefully selected greedy procedures.
For the splittable and preemptive model, we additionally have to prove that schedules with a certain simple structure exist, and in the preemptive model, the $\MCIP$ has to be extended.
In Section \ref{sec:EPTAS} the basic versions of the EPTAS are presented and in Section \ref{sec:better_running_time}, some improvements of the running time for the splittable and the setup class model are discussed.

\subparagraph*{Related work.}

For an overview on $n$-fold IPs and their applications, we refer to the book by Onn \cite{onn2010nonlinear}.
There have been recent applications of $n$-fold integer programming to scheduling problems in the context of parameterized algorithms:
Knop and Kouteck{\`y} \cite{knop2016scheduling} showed, among other things, that the problem of makespan minimization on unrelated parallel machines, where the processing times are dependent on both jobs and machines, is fixed-parameter tractable with respect to the maximum processing time and the number of distinct machine types.
This was generalized to the parameters maximum processing time and rank of the processing time matrix by Chen et al. \cite{chen2017parameterized}.
Furthermore, Knop, Kouteck{\`y} and Mnich \cite{KKM17} provided an improved algorithm for a special type of $n$-fold IPs yielding improved running times for several applications of $n$-fold IPs including results for scheduling problems.

There is extensive literature concerning scheduling problems with setup times.
We highlight a few closely related results and otherwise refer to the surveys \cite{allahverdi1999survey,allahverdi2008survey}.
In the following, we use the term $\alpha$-approximation as an abbreviation for polynomial time algorithms with approximation guarantee $\alpha$.
The setup class model was first considered by Mäcker et al. \cite{macker2015non} in the special case that all classes have the same setup time.
They designed a $2$-approximation and additionally a $3/2 + \eps$-approximation for the case that the overall length of the jobs from each class is bounded.
Jansen and Land \cite{SetupPTAS2016} presented a simple $3$-approximation with linear running time, a $2+\eps$-approximation, and the aforementioned PTAS for the general setup class model.
As indicated before, Chen et al. \cite{chen2006lot} developed a $5/3$-approximation for the splittable model.
A generalization of this, in which both setup and processing times are job and machine dependent, has been considered by Correa et al. \cite{correa2015strong}.
They achieve a $(1+\phi)$-approximation, where $\phi$ denotes the golden ratio, using a newly designed linear programming formulation.
Moreover, there are recent results concerning machine scheduling in the splittable model considering the sum of the (weighted) completion times as the objective function, e.g. \cite{schalekamp2015split,correa2016splitting}. 
For the preemptive model, a PTAS for the special case that all jobs have the same setup time has been developed by Schuurman and Woeginger \cite{schuurman1999preemptive}.
The mentioned $(4/3 +\eps)$-approximation for the general case \cite{schuurman1999preemptive} follows the same approach.
Furthermore, a combination of the setup class and the preemptive model has been considered, in which the jobs are scheduled preemptively, but the setup times are class dependent.
Monma and Potts \cite{monma1993analysis} presented, among other things, a $(2-1/(\floor{m/2}+1))$-approximation for this model, and later Chen \cite{chen1993better} achieved improvements for some special cases.

\section{Preliminaries}\label{sec:prelim}

In the following, we establish some concepts and notations, formally define the considered problems, and outline the dual approximation approach by Hochbaum and Shmoys \cite{dualapprox}, as well as $n$-fold integer programs.

For any integer $n$, we denote the set $\set{1,\dots, n}$ by $[n]$; we write $\log(\cdot)$ for the logarithm with basis $2$; and we will usually assume that some instance $I$ of the problem considered in the respective context is given together with an accuracy parameter $\eps \in (0,1)$ such that $1/\eps$ is an integer.
Furthermore for any two sets $X,Y$ we write $Y^X$ for the set of functions $f:X\rightarrow Y$. 
If $X$ is finite, we say that $Y$ is indexed by $X$ and sometimes denote the function value of $f$ for the argument $x\in X$ by $f_x$.

\subparagraph*{Problems.}

For all three of the considered models, a set $\jobs$ of $n$ jobs with processing times $p_j\in\QQ_{>0}$ for each job $j\in\jobs$ and a number of machines $m$ is given.
In the preemptive and the splittable model, the input additionally includes a setup time $s_j\in\QQ_{>0}$ for each job $j\in\jobs$; while in the setup class model, it includes a number $K$ of setup classes, a setup class $k_j\in[K]$ for each job $j\in\jobs$, as well as setup times $s_k\in\QQ_{>0}$ for each $k\in[K]$. 

We take a closer look at the definition of a schedule in the preemptive model.
The jobs may be split.
Therefore, partition sizes $\jcard: \jobs \rightarrow \ZZ_{>0}$, together with  processing time fractions $\jfrac_j:[\jcard(j)]\rightarrow(0,1]$, such that $\sum_{k\in[\jcard(j)]}\jfrac_j(k)=1$, have to be found, meaning that job $j$ is split into $\jcard(j)$ many parts and the $k$-th part for $k\in[\kappa(j)]$ has processing time $\jfrac_j(k)p_j$.
This given, we define $\jparts = \sett{(j,k)}{j\in\jobs,k\in[\jcard(j)]}$ to be the set of job parts.
Now, an assignment $\sigma:\jparts\rightarrow[m]$ along with starting times $\xi:\jparts\rightarrow\QQ_{>0}$ has to be determined, such that any two job parts assigned to the same machine or belonging to the same job do not overlap.
More precisely, we have to assure that for each two job parts $(j,k),(j',k')\in\jparts$ with $\sigma(j,k)=\sigma(j',k')$ or $j=j'$, we have $\xi(j,k)+ s_j + \lambda_j(k)p_j \leq \xi(j')$ or $\xi(j',k')+ s_{j'} + \lambda_{j'}(k)p_{j'} \leq \xi(j)$.
A schedule is given by $(\jcard,\jfrac,\sigma,\xi)$ and the makespan can be defined as $\Cmax = \max_{(j,k)\in\jparts} (\xi(j,k)+ s_j + \lambda_j(k)p_j)$.
Note that the variant of the problem in which overlap between a job part and setup of the same job is allowed is equivalent to the one presented above.
This was pointed out by Schuurmann and Woeginger \cite{schuurman1999preemptive} and can be seen with a simple swapping argument.

In the splittable model, it is not necessary to determine starting times for the job parts, because, given the assignment $\sigma$, the job parts assigned to each machine can be scheduled as soon as possible in arbitrary order without gaps.
Hence, in this case, the output is of the form $(\jcard,\jfrac,\sigma)$ and the makespan can be defined as $\Cmax = \max_{i\in[m]} \sum_{(j,k)\in\sigma^{-1}(i)} s_j + \jfrac_j(k)p_j$.

Lastly, in the setup class model the jobs are not split and given an assignment, the jobs assigned to each machine can be scheduled in batches comprised of the jobs of the same class assigned to the machine without overlaps and gaps.
The output is therefore just an assignment $\sigma:\jobs\rightarrow[m]$ and the makespan is given by $\Cmax = \max_{i\in[m]} \sum_{j\in\sigma^{-1}(i)} p_j + \sum_{k\in\sett{k_j}{j\in\sigma{-1}(i)}} s_k $.

Note, that in the preemptive and the setup class model, we can assume that the number of machines is bounded by the number of jobs:
If there are more machines than jobs, placing each job on a private machine yields an optimal schedule in both models and the remaining machines can be ignored.
This, however, is not the case in the splittable model, which causes a minor problem in the following.

\subparagraph*{Dual Approximation.}

All of the presented algorithms follow the dual approximation framework introduced by Hochbaum and Shmoys \cite{dualapprox}:
Instead of solving the minimization version of a problem directly, it suffices to find a procedure that for a given bound $T$ on the objective value either correctly reports that there is no solution with value $T$ or returns a solution with value at most $(1+ a \eps)T$ for some constant $a$.
If we have some initial upper bound $B$ for the optimal makespan $\Opt$ with $B\leq b\Opt$ for some $b$, we can define a PTAS by trying different values $T$ from the interval $[B/b,B]$ in a binary search fashion, and find a value $T^*\leq (1+\Oh(\eps))\Opt$ after $\Oh(\log b/\eps)$ iterations.
Note that for all of the considered problems constant approximation algorithms are known, and the sum of all processing and setup times is a trivial $m$-approximation.
Hence, we always assume that a target makespan $T$ is given.
Furthermore, we assume that the setup times and in the preemptive and setup class cases also the processing times are bounded by $T$, because otherwise we can reject $T$ immediately.

\subparagraph*{$n$-fold Integer Programs.}

We briefly define $n$-fold integer programs (IP) following the notation of \cite{nfoldcubic} and \cite{knop2016scheduling} and state the main algorithmic result needed in the following.
Let $n,r,s,t\in\ZZ_{>0}$ be integers and $A$ be an integer $((r + ns) \times nt)$-matrix of the following form:
\[A=
\begin{pmatrix}
  A_1    & A_1    & \cdots & A_1 \\
  A_2    & 0      & \cdots & 0 \\
  0      & A_2    & \cdots & 0 \\
  \vdots & \vdots & \ddots & \vdots \\
  0      & 0      & \cdots & A_2 
 \end{pmatrix}
\]
The matrix $A$ is the so-called $n$-fold product of the bimatrix $\binom{A_1}{A_2}$, with $A_1$ an $r\times t$ and $A_2$ an $s\times t$ matrix.
Furthermore, let $w,\ell,u\in\ZZ^{nt}$ and $b\in\ZZ^{r+ns}$.
Then the $n$-fold integer programming problem is given by:
\[\min\sett{wx}{Ax = b, \ell \leq x \leq u, x\in\ZZ^{nt}}\]
We set $\Delta$ to be the maximum absolute value occurring in $A$.
Up to recently the best known algorithm for solving $n$-fold IPs was due to Hemmecke, Onn and Romanchuk \cite{nfoldcubic}:
\begin{theorem}\label{thm:solving_n-fold_old}
Let $\varphi$ be the encoding length of $w$, $b$, $\ell$, $u$ and $\Delta$.
The $n$-fold integer programming problem can be solved in time $\Oh(\Delta^{3t(rs+st+r+s)}n^3\varphi)$, when $r$, $s$ and $t$ are fixed.
\end{theorem}
However, in 2018 both Eisenbrand, Hunkenschröder and Klein \cite{KLS-n-fold-2018} and independently Koutecký, Levin and Onn \cite{KLS-n-fold-2018} developed algorithms with improved and very similar running times.
We state a variant due to Eisenbrand et al. that is adapted to our needs:
\begin{theorem}\label{thm:solving_n-fold}
Let $\varphi$ be the encoding length of the largest number occurring in the input, and $\Phi = \max_i(u_i-\ell_i)$.
The $n$-fold integer programming problem can be solved in time $(rs\Delta)^{\Oh(r^2s+rs^2)}t^2n^2\varphi\log(\Phi)\log(nt\Phi)$.
\end{theorem}

The variables $x$ can naturally be partitioned into \emph{bricks} $x^{(q)}$ of dimension $t$ for each $q\in [n]$, such that $x=(x^{(1)},\dots x^{(n)})$.
Furthermore, we denote the constraints corresponding to $A_1$ as \emph{globally uniform} and the ones corresponding to $A_2$ as \emph{locally uniform}.
Hence, $r$ is the number of globally and $s$ the number of locally uniform constraints (ignoring their $n$-fold duplication); $t$ the \emph{brick size} and $n$ the \emph{brick number}.

\section{Module Configuration IP}\label{sec:ILP}

In this section, we state the configuration IP for machine scheduling; introduce a basic version of the module configuration IP (MCIP) that is already sufficiently general to work for both the splittable and setup class model; and lastly show that the configuration IP can be expressed by the $\MCIP$ in multiple ways.
Before that, however, we formally introduce the concept of \emph{configurations}.

Given a set of objects $A$, a configuration $C$ of these objects is a vector of multiplicities indexed by the objects, i.e., $C\in\ZZ_{\geq 0}^A$.
For given sizes $\size(a)$ of the objects $a\in A$, the size $\size(C)$ of a configuration $C$ is defined as $\sum_{a\in A}C_a\size(a)$. 
Moreover, for a given bound $B$, we define $\confs_A(B)$ to be the set of configurations of $A$ that are bounded in size by $B$, that is, $\confs_A(B) = \sett{C\in\ZZ_{\geq 0}^A}{\size(C)\leq B}$.

\subparagraph*{Configuration IP.}

We give a recollection of the configuration IP for scheduling on identical parallel machines.
Let $P$ be the set of distinct processing times for some instance $I$ with multiplicities $n_p$ for each $p\in P$, meaning, $I$ includes exactly $n_p$ jobs with processing time $p$.
The size $\size(p)$ of a processing time $p$ is given by itself.
Furthermore, let $T$ be a guess of the optimal makespan.
The configuration IP for $I$ and $T$ is given by variables $x_C\geq 0$ for each $C\in \confs_P(T)$ and the following constraints:
\begin{align}
\sum_{C\in\confs_P(T)}x_{C} & = m & \label{eq:CILP_machs} \\ 
\sum_{C\in\confs_P(T)} C_p x_{C} & = n_p & \forall p\in P\label{eq:CILP_jobs} 
\end{align}
Due to constraint (\ref{eq:CILP_machs}), exactly one configuration is chosen for each machine, while (\ref{eq:CILP_jobs}) ensures that the correct number of jobs or job sizes is covered.

\subparagraph*{Module Configuration IP.}

Let $\bobs$ be a set of basic objects (e.g. jobs or setup classes) and let there be $D$ integer values $B_1,\dots,B_D$ for each basic object $B\in\bobs$ (e.g. processing time or numbers of different kinds of jobs).
Our approach is to cover the basic objects with so-called \emph{modules} and in turn cover the modules with configurations.
Depending on the context, modules correspond to batches of jobs or job piece sizes together with a setup time and can also encompass additional information like a starting time.
Let $\mods$ be a set of such modules.
In order to cover the basic objects, each module $M\in\mods$ also has $D$ integer values $M_1,\dots,M_D$.
Furthermore, each module $M$ has a size $\size(M)$ and a set of eligible basic objects $ \bobs(M)$.
The latter is needed because not all modules are compatible with all basic objects, e.g., because they do not have the right setup times.
The configurations are used to cover the modules, however, it typically does not matter which module exactly is covered, but rather which size the module has.
Let $\msizes$ be the set of distinct module sizes, i.e., $\msizes = \sett{\size(M)}{M\in\mods}$, and for each module size $h\in\msizes$ let $\mods(h)$ be the set of modules with size $h$.
We consider the set $\confs$ of configurations of module sizes which are bounded in size by a guess of the makespan $T$, i.e., $\confs=\confs_\msizes(T)$.
In the preemptive case configurations need to additionally encompass information about starting times of modules, and therefore the definition of configurations will be slightly more complicated in that case.

Since we want to chose configurations for each machine, we have variables $x_C$ for each $C\in\confs$ and constraints corresponding to (\ref{eq:CILP_machs}).
Furthermore, we chose modules with variables $y_M$ for each $M\in\mods$ and because we want to cover the chosen modules with configurations, we have some analogue of constraint (\ref{eq:CILP_jobs}), say $\sum_{C\in\confs(T)} C_{h} x_{C}  = \sum_{M\in \mods(h)}y_M $ for each module size $h \in\msizes$.
It turns out however, that to properly cover the basic objects with modules, we need the variables $y_M$ for each basic object, and this is were $n$-fold IPs come into play: 
The variables stated so far form a brick of the variables of the $n$-fold IP and there is one brick for each basic object, that is, we have, for each $B\in\bobs$, variables $x^{(B)}_C$ for each $C\in\confs$, and $y^{(B)}_M$ for each $M\in\mods$.
Using the upper bounds of the $n$-fold model, variables $y^{(B)}_M$ are set to zero, if $B$ is not eligible for $M$; and we set the lower bounds of all variables to zero.
Sensible upper bounds for the remaining variables, will be typically clear from context.
Besides that, the module configuration integer program $\MCIP$ (for $\bobs$, $\mods$ and $\confs$) is given by:
\begin{align}
\sum_{B\in\bobs}\sum_{C\in\confs}x^{(B)}_{C} & = m & \label{eq:MCIP_machs} \\ 
\sum_{B\in\bobs}\sum_{C\in\confs(T)} C_h x^{(B)}_{C} & = \sum_{B\in\bobs}\sum_{M\in \mods(h)}y^{(B)}_M & \forall h \in\msizes \label{eq:MCIP_mods}\\
\sum_{M\in\mods}M_d y^{(B)}_M & = B_d & \forall B\in\bobs, d\in[D] \label{eq:MCIP_bobs}
\end{align}
It is easy to see that the constraints (\ref{eq:MCIP_machs}) and (\ref{eq:MCIP_mods}) are globally uniform.
They are the mentioned adaptations of (\ref{eq:CILP_machs}) and (\ref{eq:CILP_jobs}).
The constraint (\ref{eq:MCIP_bobs}), on the other hand, is locally uniform and ensures that the basic objects are covered.

Note that, while the duplication of the configuration variables does not carry meaning, it also does not upset the model: 
Consider the modified $\MCIP$ that is given by not duplicating the configuration variables.
A solution $(\tilde{x}, \tilde{y})$ for this IP gives a solution $(x,y)$ for the $\MCIP$ by fixing some basic object $B^*$, setting $x^{(B^*)}_C = \tilde{x}_C$ for each configuration $C$, setting the remaining configuration variables to $0$, and copying the remaining variables.
Given a solution $(x,y)$ for the $\MCIP$, on the other hand, gives a solution for the modified version $(\tilde{x}, \tilde{y})$ by setting $\tilde{x}_C = \sum_{B\in\bobs}x^B_C$ for each configuration $C$.
Summarizing we get:

\begin{observation}\label{rem:MCIP}
The $\MCIP$ is an $n$-fold IP with brick-size $t = |\mods| + |\confs|$, brick number $n = |\bobs|$, $r = |\msizes|+1$ globally uniform and $s = D$ locally uniform constraints.
\end{observation}

Moreover, in all the considered applications we will minimize the overall size of the configurations, i.e., $\sum_{B\in\bobs}\sum_{C\in\confs}\size(C)x^{(B)}_{C}$.
This will be required, because in the simplification steps of our algorithms some jobs are removed and have to be reinserted later, and we therefore have to make sure that no space is wasted.

\subparagraph*{First Example.}

We conclude the section by pointing out several different ways to replace the classical configuration IP for scheduling on identical machines with the $\MCIP$, thereby giving some intuition for the model.
The first possibility is to consider the jobs as the basic objects and their processing times as their single value ($\bobs= \jobs$, $D=1$); the modules are the processing times ($\mods = P$), and a job is eligible for a module, if its processing time matches; and the configurations are all the configurations bounded in size by $T$.
Another option is to chose the processing times as basic objects, keeping all the other definitions essentially like before.
Lastly, we could consider the whole set of jobs or the whole set of processing times as a single basic object with $D=|P|$ different values.
In this case, we can define the set of modules as the set of configurations of processing times bounded by $T$.

\section{EPTAS results}\label{sec:EPTAS}

In this section, we present approximation schemes for each of the three considered problems.
Each of the results follows the same approach:
The instance is carefully simplified, a schedule for the simplified instance is found using the $\MCIP$, and this schedule is transformed into a schedule for the original instance.
The presentation of the result is also similar for each problem:
We first discuss how the instance can be sensibly simplified, and how a schedule for the simplified instance can be transformed into a schedule for the original one.
Next, we discuss how a schedule for the simplified instance can be found using the $\MCIP$, and lastly, we summarize and analyze the taken steps.

For the sake of clarity, we have given rather formal definitions for the problems at hand in Section \ref{sec:prelim}.
In the following, however, we will use the terms in a more intuitive fashion for the most part, and we will, for instance, often take a geometric rather than a temporal view on schedules and talk about the \emph{length} or the \emph{space} taken up by jobs and setups on machines rather than time.
In particular, given a schedule for an instance of any one of the three problems together with an upper bound for the makespan $T$, the \emph{free space} with respect to $T$ on a machine is defined as the summed up lengths of time intervals between $0$ and $T$ in which the machine is idle.
The free space (with respect to $T$) is the summed up free space of all the machines.
For bounds $T$ and $L$ for the makespan and the free space, we say that a schedule is a $(T,L)$-schedule if its makespan is at most $T$ and the free space with respect to $T$ is at least $L$.

When transforming the instance we will increase or decrease processing and setup times and fill in or remove extra jobs.
Consider a $(T',L')$-schedule, where $T'$ and $L'$ denote some arbitrary makespan or free space bounds.
If we fill in extra jobs or increase processing or setup times, but can bound the increase on each machine by some bound $b$, we end up with a $(T' + b,L')$-schedule for the transformed instance.
In particular we have the same bound for the free space, because we properly increased the makespan bound.
If, on the other hand, jobs are removed or setup times decreased, we obviously still have a $(T',L')$-schedule for the transformed instance.
This will be used frequently in the following.

\subsection{Setup Class Model}

We start with the setup class model.
In this case, we can essentially reuse the simplification steps that were developed by Jansen and Land \cite{SetupPTAS2016} for their PTAS.
The main difference between the two procedures is that we solve the simplified instance via the $\MCIP$, while they used a dynamic program.
For the sake of self-containment, we include our own simplification steps, but remark that they are strongly inspired by those from \cite{SetupPTAS2016}.
In Section \ref{sec:better_running_time} we give a more elaborate rounding procedure resulting in an improved running time.

\subparagraph*{Simplification of the Instance.}

In the following, we distinguish \emph{big setup} jobs $j$ jobs belonging to classes $k$ with setup times $s_k \geq \eps^3 T$ and \emph{small setup} jobs with $s_k < \eps^3 T$.
We denote the corresponding subsets of jobs by $\bsjobs$ and $\ssjobs$ respectively.
Furthermore, we call a job \emph{tiny} or \emph{small}, if its processing time is smaller than $\eps^4 T$ or $\eps T$ respectively, and \emph{big} or \emph{large} otherwise.
For any given set of jobs $J$, we denote the subset of tiny jobs from $J$ with $J_\ptiny$ and the small, big and large jobs analogously.
We simplify the instance in four steps, aiming for an instance that exclusively includes big jobs with big setup times and additionally only a constant number of distinct processing and setup times. 
For technical reason we assume $\eps \leq 1/2$.

We proceed with the first simplification step.
Let $I_1$ be the instance given by the job set $\jobs\setminus\ssjobs_\psmall$ and $Q$ the set of setup classes completely contained in $\ssjobs_\psmall$, i.e., $Q = \sett{k}{\forall j\in\jobs: k_j = k \Rightarrow j\in\ssjobs_\psmall }$.
An obvious lower bound on the space taken up by the jobs from $\ssjobs_\psmall$ in any schedule is given by $L=\sum_{j\in\ssjobs_\psmall} p_j + \sum_{k\in Q} s_k$.
Note that the instance $I_1$ may include a reduced number $K'$ of setup classes.
\begin{lemma}\label{lem:sclass_rounding1}
A schedule for $I$ with makespan $T$ induces a $(T,L)$-schedule for $I_{1}$, that is, a schedule with makespan $T$ and free space at least $L$; and any $(T',L)$-schedule for $I_1$ can be transformed into a schedule for $I$ with makespan at most $(1+\eps)T'+ 2\eps^3 T$.
\end{lemma}
\begin{proof}
The first claim is obvious and we therefore assume that we have a $(T',L)$-schedule for $I_1$.
We group the jobs from $\ssjobs_\psmall$ by setup classes and first consider the groups with summed up processing time at most $\eps^2 T$.
For each of these groups we check whether the respective setup class contains a large job.
If this is the case, we schedule the complete group on a machine on which such a large job is already scheduled if possible using up free space.
Since the large jobs have a length of at least $\eps T$, there are at most $T'/(\eps T)$ many large jobs on each machine and therefore the schedule on the respective machine has length at most $(1+\eps)T'$ or there is free space with respect to $T'$ left.
If, on the other hand, the respective class does not contain a large job and is therefore fully contained in $\ssjobs_\psmall$, we create a container including the whole class and its setup time.
Note that the overall length of the container is at most $(\eps^2+\eps^3)T\leq \eps T$ (using $\eps \leq 1/2$).
Next, we create a sequence containing the containers and the remaining jobs ordered by setup class.
We insert the items from this sequence greedily into the remaining free space in a next-fit fashion, exceeding $T'$ on each machine by at most one item from the sequence.
This can be done because we had a free space of at least $L$ and the inserted objects had an overall length of at most $L$.
To make the resulting schedule feasible, we have to insert some setup times.
However, because the overall length of the jobs from each class in need of a setup is at least $\eps^2 T$ and the sequence was ordered by classes, there are at most $T'/(\eps^2 T) + 2$ distinct classes without a setup time on each machine.
Inserting the missing setup times will therefore increase the makespan by at most $(T'/(\eps^2 T) + 2)\eps^3 T = \eps T' + 2\eps^3 T$.
\end{proof}

Next, we deal with the remaining (large) jobs with small setup times $j\in\ssjobs_\plarge$.
Let $I_2$ be the instance we get by increasing the setup times of the classes with small setup times to $\eps^3 T$.
We denote the setup time of class $k\in[K']$ for $I_2$ by $s'_k$.
Note that there are no small setup jobs in $I_2$.
\begin{lemma}\label{lem:sclass_rounding2}
A $(T',L')$-schedule $I_1$ induces a $((1 + \eps^2)T',L')$-schedule for $I_2$, and a $(T',L')$-schedule for $I_2$ is also a $(T',L')$-schedule for $I_1$.
\end{lemma}
\begin{proof}
The first claim is true because in a schedule with makespan at most $T$ there can be at most $T'/(\eps T)$ many large jobs on any machine, and the second claim is obvious.
\end{proof}

Let $I_3$ be the instance we get by replacing the jobs from $\bsjobs_\ptiny$ with placeholders of size $\eps^4 T$.
More precisely, for each class $k\in[K]$ we introduce $\ceil{(\sum_{j\in\bsjobs_\ptiny,k_j=k}p_j)/(\eps^4 T)}$ many jobs with processing time $\eps^4 T$ and class $k$.
We denote the job set of $I_3$ by $\jobs'$ and the processing time of a job $j\in\jobs'$ by $p'_j$.
Note that $I_3$ exclusively contains big jobs with big setup times.
\begin{lemma}\label{lem:sclass_rounding3}
If there is a $(T',L')$-schedule for $I_2$, there is also a $((1 + \eps)T',L')$-schedule; and if there is a $(T',L')$-schedule for $I_3$, there is also a $((1 + \eps)T',L')$-schedule for $I_2$.
\end{lemma}
\begin{proof}
Note, that for any $(T',L')$-schedule for $I_2$ or $I_3$ there are at most $T'/ (\eps^3T) $ many distinct big setup classes scheduled on any machine.
Hence, when considering such a schedule for $I_2$, we can remove the tiny jobs belonging to $\bsjobs_\ptiny$ from the machines and instead fill in the placeholders, such that each machine for each class receives at most as much length from that class, as was removed, rounded up to the next multiple of $\eps^4 T$.
All placeholders can be placed like this and the makespan is increased by at most $(T'/ (\eps^3T))\eps^4 T = \eps T'$.
If, on the other hand, we consider such a schedule for $I_3$, we can remove the placeholders and instead fill in the respective tiny jobs, again overfilling by at most one job.
This yields a $((1 + \eps)T',L')$-schedule for $I_2$ with the same argument.
\end{proof}

Lastly, we perform both a geometric and an arithmetic rounding step for the processing and setup times.
The geometric rounding is needed to suitably bound the number of distinct processing and setup times and due to the arithmetic rounding we will be able to guarantee integral coefficients in the IP.
More precisely, we set $\tilde{p}_j = (1+\eps)^{\ceil{\log_{1+\eps}p'_j/(\eps^4 T)}}\eps^4 T$ and $\bar{p}_j = \ceil{\tilde{p}_j/\eps^5 T}\eps^5 T$ for each $j\in\jobs'$, as well as $\tilde{s}_j = (1+\eps)^{\ceil{\log_{1+\eps}s'_j/(\eps^3 T)}}\eps^3 T$ and $\bar{s}_k = \ceil{\tilde{s}_j/\eps^5 T}\eps^5 T$ for each setup class $k\in[K']$.
The resulting instance is called $I_4$.
\begin{lemma}\label{lem:sclass_rounding4}
A $(T',L')$-schedule for $I_3$ induces a $((1 + 3\eps)T',L')$-schedule for $I_4$, and any $(T',L')$-schedule for $I_4$ can be turned into a $(T',L')$-schedule for $I_3$.
\end{lemma}
\begin{proof}
For the first claim, we first stretch a given schedule by $(1+\eps)$.
This enables us to use the processing and setup times due to the geometric rounding step.
Now, using the ones due to the second step increases the schedule by at most $2\eps T'$, because there where at most $T'/(\eps^4 T)$ many big jobs on any machine to begin with.
The second claim is obvious.
\end{proof}
Based on the rounding steps, we define two makespan bounds $\bar{T}$ and $\breve{T}$:
Let $\bar{T}$ be the makespan bound that is obtained from $T$ by the application of the Lemmata \ref{lem:sclass_rounding1}-\ref{lem:sclass_rounding4} in sequence, i.e., $\bar{T} = (1 + \eps^2)(1 + \eps)(1 + 3\eps) T = (1+\Oh(\eps))T$.
We will find a $(\bar{T},L)$-schedule for $I_4$ utilizing the $\MCIP$ and afterward apply the Lemmata \ref{lem:sclass_rounding1}-\ref{lem:sclass_rounding4} backwards, to get a schedule with makespan $\breve{T} = (1+\eps)^2\bar{T} + \eps^3 T= (1+\Oh(\eps))T$.

Let $P$ and $S$ be the sets of distinct occurring processing and setup times for instance $I_4$.
Because of the rounding, the minimum and maximum lengths of the setup and processing times, and $\eps < 1$, we can bound $|P|$ and $|S|$ by $\Oh(\log_{1+\eps} 1/\eps)=\Oh(1/\eps\log 1/\eps)$.

\subparagraph*{Utilization of the MCIP.}

At this point, we can employ the module configuration IP.
The basic objects in this context are the setup classes, i.e., $\bobs=[K']$, and the different values are the numbers of jobs with a certain processing time, i.e., $D = |P|$.
We set $n_{k,p}$ to be the number of jobs from setup class $k\in[K']$ with processing time $p\in P$.
The modules correspond to batches of jobs together with a setup time.
Batches of jobs can be modeled as configurations of processing times, that is, multiplicity vectors indexed by the processing times.
Hence, we define the set of modules $\mods$ to be the set of pairs of configurations of processing times and setup times with a summed up size bounded by $\bar{T}$, i.e., $\mods = \sett{(C,s)}{C\in\confs_P(\bar{T}),s\in S, s + \size(C)\leq \bar{T}}$, and write $M_p = C_p$ and $s_M = s$ for each module $M=(C,s)\in\mods$.
The values of a module $M$ are given by the numbers $M_p$ and its size $\size(M)$ by $s_M + \sum_{p\in P}M_p p$.
Remember that the configurations $\confs$ are the configurations of module sizes $\msizes$ that are bounded in size by $\bar{T}$, i.e., $\confs=\confs_\msizes(\bar{T})$.
A setup class is eligible for a module, if the setup times fit, i.e., $\bobs_M = \sett{k\in[K']}{s_k=s_M}$.
Lastly, we establish $\eps^5 T=1$ by scaling.

For the sake of readability, we state the resulting constraints of the $\MCIP$ with adapted notation and without duplication of the configuration variables:
\begin{align}
\sum_{C\in\confs}x_{C} & = m & \label{eq:MCIP_machs_class} \\ 
\sum_{C\in\confs} C_h x_{C} & = \sum_{k\in[K']}\sum_{M\in \mods(h)}y^{(k)}_M & \forall h\in\msizes \label{eq:MCIP_mods_class}\\
\sum_{M\in\mods}M_p y^{(k)}_M & = n_{k,p} & \forall k\in[K'],p\in P \label{eq:MCIP_jobs_class}
\end{align}
Note that the coefficients are all integral and this includes those of the objective function, i.e., $\sum_C \size(C)x_C $, because of the scaling step.
\begin{lemma}\label{lem:class_IP_to_sched}
With the above definitions, there is a $(\bar{T},L)$-schedule for $I_4$, iff the $\MCIP$ has a solution with objective value at most $m\bar{T} - L$.
\end{lemma}
\begin{proof}
Let there be a  $(\bar{T},L)$-schedule for $I_4$.
Then the schedule on a given machine corresponds to a distinct configuration $C$ that can be determined by counting for each possible group size $a$ the batches of jobs from the same class whose length together with the setup time adds up to an overall length of $a$.
Note that the length of this configuration is equal to the used up space on that machine.
We fix an arbitrary setup class $k$ and set the variables $x^{(k)}_C$ accordingly (and $x^{(k')}_C=0$ for $k'\neq k$ and $C\in\confs$).
By this setting we get an objective value of at most $m\bar{T} - L$ because there was $L$ free space in the schedule.
For each class $k$ and module $M$, we count the number of machines on which the there are exactly $M_p$ jobs with processing time $p$ from class $k$ for each $p\in P$, and set $y^{(k)}_M$ accordingly.
It is easy to see that the constraints are satisfied by these definitions.

Given a solution $(x,y)$ of the $\MCIP$, we define a corresponding schedule:
Because of (\ref{eq:MCIP_machs_class}) we can match the machines to configurations such that each machine is matched to exactly one configuration.
If machine $i$ is matched to $C$, for each group $G$ we create $C_G$ slots of length $\size(G)$ on $i$.
Next, we divide the setup classes into batches.
For each class $k$ and module $M$, we create $y^{(k)}_M$ batches of jobs from class $k$ with $M_p$ jobs with processing time $p$ for each $p\in P$ and place the batch together with the corresponding setup time into a fitting slot on some machine.
Because of (\ref{eq:MCIP_jobs_class}) and (\ref{eq:MCIP_mods_class}) all jobs can be placed by this process.
Note that the used space equals the overall size of the configurations and we therefore have free space of at least $L$.
\end{proof}

\subparagraph*{Result.}

Using the above results, we can formulate and analyze the following procedure:
\begin{algorithm}\
\begin{enumerate}
\item Generate the modified instance $I_4$: 
\begin{itemize}
\item Remove the small jobs with small setup times.
\item Increase the setup times of the remaining classes with small setup times.
\item Replace the tiny jobs with big setup times.
\item Round up the resulting processing and setup times. 
\end{itemize}
\item Build and solve the $\MCIP$ for $I_4$.
\item If the $\MCIP$ is infeasible, or the objective value greater than $m \bar{T} - L$, report that $I$ has no solution with makespan $T$.
\item Otherwise build the schedule with makespan $\bar{T}$ and free space at least $L$ for $I_4$. 
\item Transform the schedule into a schedule for $I$ with makespan at most $\breve{T}$:
\begin{itemize}
\item Use the prerounding processing and setup times. 
\item Replace the placeholders by the tiny jobs with big setup times.
\item Use the orignal setup times of the classes with small setup times.
\item Insert the small jobs with small setup times into the free space.
\end{itemize}
\end{enumerate}
\end{algorithm}
The procedure is correct due to the above results.
To analyze its running time, we first bound the parameters of the $\MCIP$.
We have $|\bobs| = K' \leq K$ and $D = |P|$ by definition, 
and $|\mods| = \Oh(|S|(1/\eps^3)^{|P|}) = 2^{\Oh(\nicefrac{1}{\eps}\log^2\nicefrac{1}{\eps})}$, because $|S|,|P|\in\Oh(1/\eps\log 1/\eps)$.
This is true, due to the last rounding step, which also implies $|\msizes|\in\Oh(1/\eps^5)$, yielding $|\confs| = |\msizes|^{\Oh(1/\eps^3)} = 2^{\Oh(\nicefrac{1}{\eps^3}\log\nicefrac{1}{\eps})}$.
According to Observation \ref{rem:MCIP}, this yields a brick size of $t = 2^{\Oh(\nicefrac{1}{\eps^3}\log\nicefrac{1}{\eps})}$, a brick number of $K$, $\Oh(1/\eps^5)$ globally, and $\Oh(1/\eps\log 1/\eps)$ locally uniform constraints for the $\MCIP$.
We have $\Delta =\Oh(1/\eps^5)$, because all occurring values in the processing time matrix are bounded in $\bar{T}$, and we have $\bar{T} =\Oh(1/\eps^5)$, due to the scaling.
Furthermore, the values of the objective function, the right hand side, and the upper and lower bounds on the variables are bounded by $\Oh(n/\eps^5)$, yielding a bound of $\Oh(\log n/\eps^5)$ for the encoding length of the biggest number in the input $\varphi$.
Lastly, all variables can be bounded by $0$ from below and $\Oh(m/\eps^3)$ from above, yielding $\Phi = \Oh(m/\eps^3)$.

By Theorem \ref{thm:solving_n-fold} and some arithmetic, the $\MCIP$ can be solved in time:
\[(rs\Delta)^{\Oh(r^2s+rs^2)}t^2n^2\varphi\log(\Phi)\log(nt\Phi) = 2^{\Oh(\nicefrac{1}{\eps^{11}}\log^2\nicefrac{1}{\eps})}K^2\log n\log m \log Km\]
When building the actual schedule, we iterate through the jobs and machines like indicated in the proof of Lemma \ref{lem:class_IP_to_sched}, yielding the following:
\begin{theorem}
The algorithm for the setup class model finds a schedule with makespan $(1+\Oh(\eps))T$ or correctly determines that there is no schedule with makespan $T$ in time
$2^{\Oh(\nicefrac{1}{\eps^{11}}\log^2\nicefrac{1}{\eps})}K^2nm\log Km$.
\end{theorem}

\subsection{Splittable Model}\label{sec:EPTAS_slittable}

The approximation scheme for the splittable model presented in this section is probably the easiest one discussed in this work.
There is, however, one problem concerning this procedure:
Its running time is polynomial in the number of machines, which might be exponential in the input size. 
In Section \ref{sec:better_running_time} we show how this problem can be overcome and further improve the running time.

\subparagraph*{Simplification of the Instance.}

In this context the set of big setup jobs $\bsjobs$ is given by the jobs with setup times at least $\eps T$ and the small setup jobs $\ssjobs$ are all the others.
Let $L = \sum_{j\in \ssjobs} (s_j + p_j)$.
Because every job has to be scheduled and every setup has to be paid at least once, $L$ is a lower bound on the summed up space due to small jobs in any schedule. 
Let $I_1$ be the instance that we get by removing all the small setup jobs from the given instance $I$.

\begin{lemma}\label{lem:split_rounding1}
A schedule with makespan $T$ for $I$ induces a $(T,L)$-schedule for $I_1$; and any $(T',L)$-schedule for $I_1$ can be transformed into a schedule for $I$ with makespan at most $T'+\eps T$.
\end{lemma}
\begin{proof}
The first claim is obvious.
Hence, consider a sequence consisting of the jobs from $\ssjobs$ together with their set up times, where the setup up time of a job is the direct predecessor of the job.
We insert the setup times and jobs from this sequence greedily into the schedule in a next-fit fashion:
Given a machine we keep inserting the items from the sequence on the machine at the end of the schedule until the taken up space on the machine reaches $T'$.
If the current item does not fit exactly, we cut it, such that the used space on the machine is exactly $T'$.
Then we continue with the next machine. 
We can place the whole sequence like this without exceeding the makespan $T'$, because we have free space of at least $L$ which is the summed up length of the items in the sequence.
Next, we remove each setup time that was placed only partly on a machine together witch those that were placed at the end of the schedule, and insert a fitting setup time for the jobs that were scheduled without one, which can happen only once for each machine.
This yields a feasible schedule, whose makespan is increased by at most $\eps T$.   
\end{proof}
Next, we round up the processing times of $I_1$ to the next multiple of $\eps^2 T$, that is, for each job $j\in\jobs$ we set $\bar{p}_j = \ceil{p_j/(\eps^2 T)} \eps^2 T$ and $\bar{s}_j = \ceil{s_j/(\eps^2 T)} \eps^2 T$.
We call the resulting instance $I_2$, and denote its job set by $\jobs'$.

\begin{lemma}\label{lem:split_rounding2}
If there is a $(T,L')$-schedule for $I_1$, there is also a $((1 + 2\eps)T,L')$-schedule for $I_2$ in which the length of each job part is a multiple of $\eps^2 T$, and any $(T',L')$-schedule for $I_2$ yields a $(T',L')$-schedule for $I_1$.
\end{lemma}
\begin{proof}
Consider a $(T,L)$-schedule for $I_1$. 
There are at most $1/\eps$ jobs scheduled on each machine, since each setup time has a length of at least $\eps T$.
On each machine, we extend each occurring setup time and the processing time of each occurring job part by at most $\eps^2T$ to round it to a multiple of $\eps^2T$. 
This step extends the makespan by at most $2\eps T$. 
Since now each job part is a multiple of $\eps^2 T$, the total processing time of the job is a multiple of $\eps^2 T$ too. 

In the last step, we check for each job $j \in \bsjobs$ if the total processing time is now larger than the smallest multiple of $\eps^2T$, which is larger than its original processing time. 
If this is the case, we discard the spare processing time. 
Lastly, there is at least as much free space in the resulting schedule as in the original one, because we properly increased the makespan bound.
The second claim is obvious.
\end{proof}
Based on the two Lemmata, we define two makespan bounds $\bar{T} = (1 + 2\eps)T$ and $\breve{T} = \bar{T} + \eps T = (1+3\eps)T$.
We will use the $\MCIP$ to find a $(\bar{T},L)$-schedule for $I_2$ in which the length of each job part is a multiple of $\eps^2 T$.
Using the two Lemmata, this will yield a schedule with makespan at most $\breve{T}$ for the original instance $I$.

\subparagraph*{Utilization of the MCIP.}

The basic objects in this context are the (big setup) jobs, i.e., $\bobs = \bsjobs = \jobs'$, and they have only one value ($D = 1$), namely, their processing time.
Moreover, the modules are defined as the set of pairs of job piece sizes and setup times, i.e., $\mods = \sett[\big]{(q,s)}{s,q\in\sett{x\eps^2 T}{x\in\ZZ, 0<x\leq 1/\eps^2}, s\geq \eps T}$, and we write $s_M = s$ and $q_M = q$ for each module $M = (q,s)\in \mods$.
Corresponding to the value of the basic objects the value of a module $M$ is $q_M$, and its size $\size(M)$ is given by $q_M+s_M$.
A job is eligible for a module, if the setup times fit, i.e., $\bobs_M = \sett{j\in\jobs'}{s_j=s_M}$.
In order to ensure integral values, we establish $\eps^2 T=1$ via a simple scaling step.
The set of configurations $\confs$ is comprised of all configurations of module sizes $\msizes$ that are bounded in size by $\bar{T}$, i.e., $\confs=\confs_\mods(\bar{T})$.
We state the constraints of the $\MCIP$ for the above definitions with adapted notation and without duplication of the configuration variables:
\begin{align}
\sum_{C\in\confs}x_{C} & = m & \label{eq:MCIP_machs_split} \\ 
\sum_{C\in\confs} C_h x_{C} & = \sum_{j\in\jobs'}\sum_{M\in \mods(h)}y^{(j)}_M & \forall h\in\msizes \label{eq:MCIP_mods_split}\\
\sum_{M\in\mods}q_M y^{(j)}_M & = p_j & \forall j\in\jobs' \label{eq:MCIP_jobs_split}
\end{align}
Note that we additionally minimize the summed up size of the configurations, via the objective function $\sum_C \size(C)x_C$.
\begin{lemma}\label{lem:split_IP_to_sched}
With the above definitions, there is a ($\bar{T}, L$)-schedule for $I_2$ in which the length of each job piece is a multiple of $ \eps^2 T$, iff $\MCIP$ has a solution with objective value at most $m\bar{T} - L$.
\end{lemma}
\begin{proof}
Given such a schedule for $I_2$, the schedule on each machine corresponds to exactly one configuration $G$ that can be derived by counting the job pieces and setup times with the same summed up length $a$ and setting $C_G$ accordingly, where $G$ is the group of modules with length $a$.
The size of the configuration $C$ is equal to the used space on the respective machine.
Therefore, we can fix some arbitrary job $j$ and set the variables $x^{(j)}_C$ to the number of machines whose schedule corresponds to $C$ (and $x^{(j')}_C=0$ for $j'\neq j$ and $C\in\confs$).
Since there is at least a free space of $L$ for the schedule, the objective value is bounded by $m\bar{T} - L$. 
Furthermore, for each job $j$ and job part length $q$, we count the number of times a piece of $j$ with length $q$ is scheduled and set $y^{(j)}_{(q,s_j)}$ accordingly.
It is easy to see that the constraints are satisfied.

Now, let $(x,y)$ be a solution to the $\MCIP$ with objective value at most $m\bar{T} - L$.
We use the solution to construct a schedule:
For job $j$ and configuration $C$ we reserve $x^{(j)}_{C}$ machines.
On each of these machines we create $C_h$ slots of length $h$, for each module size $h\in\msizes$.
Note that because of (\ref{eq:MCIP_machs_split}) there is the exact right number of machines for this.
Next, consider each job $j$ and possible job part length $q$ and create $y^{(j)}_{(q,s_j)}$ split pieces of length $q$ and place them together with a setup of $s_j$ into a slot of length $s_j+q$ on any machine.
Because of (\ref{eq:MCIP_jobs_split}) the entire job is split up by this, and because of (\ref{eq:MCIP_mods_split}) there are enough slots for all the job pieces.
Note that the used space in the created schedule is equal to the objective value of $(x,y)$ and therefore there is at least $L$ free space.
\end{proof}

\subparagraph*{Result.}

Summing up, we can find a schedule of length at most $(1 + 3\eps) T$ or correctly determine that there is no schedule of length $T$ with the following procedure:

\begin{algorithm}
\ 
\begin{enumerate}
\item Generate the modified instance $I_2$:
\begin{itemize}
\item Remove the small setup jobs.
\item Round the setup and processing times of the remaining jobs.
\end{itemize}
\item Build and solve the $\MCIP$ for this case.
\item If the IP is infeasible, or the objective value greater than $m\bar{T} - L$, report that $I$ has no solution with makespan $T$.
\item Otherwise build the schedule with makespan $\bar{T}$ and free space at least $L$ for $\bar{I}$.
\item Transform the schedule into a schedule for $I$ with makespan at most $\breve{T}$:
\begin{itemize}
\item Use the original processing and setup times.
\item Greedily insert the small setup jobs.
\end{itemize}
\end{enumerate}
\end{algorithm}
To assess the running time of the procedure, we mainly need to bound the parameters of the $\MCIP$, namely $|\bobs|$, $|\msizes|$, $|\mods|$, $|\confs|$ and $D$.
By definition, we have $|\bobs| = |\jobs'| \leq n$ and $D=1$.
Since all setup times and job piece lengths are multiples of $\eps^2 T$ and bounded by $T$, we have $|\mods| = \Oh(1/\eps^4)$ and $|\msizes| = \Oh(1/\eps^2)$.
This yields $|\confs| \leq |\msizes|^{\Oh(1/\eps + 2)} = 2^{\Oh(\nicefrac{1}{\eps}\log \nicefrac{1}{\eps})}$, because the size of each module is at least $\eps T$ and the size of the configurations bounded by $(1+2\eps)T$.

According to Observation \ref{rem:MCIP}, we now have brick-size $t = 2^{\Oh(\nicefrac{1}{\eps}\log \nicefrac{1}{\eps})}$, brick number $|\bobs| = n$, $r = |\groups|+1 = \Oh(1/\eps^2)$ globally uniform and $s = D = 1$ locally uniform constraints.
Because of the scaling step, all occurring numbers in the constraint matrix of the $\MCIP$ are bounded by $1/\eps^2$ and therefore $\Delta \leq 1/\eps^2$.
Furthermore, each occurring number can be bounded by $\Oh(m/\eps^2)$ and this is an upper bound for each variable as well, yielding $\varphi = \Oh(\log m/\eps^2 )$ and $\Phi = \Oh(m/\eps^2)$.
Hence the $\MCIP$, can be solved in time: \[(rs\Delta)^{\Oh(r^2s+rs^2)}t^2n^2\varphi\log(\Phi)\log(nt\Phi) = 2^{\Oh(\nicefrac{1}{\eps^4}\log\nicefrac{1}{\eps})}n^2\log^2 m\log nm\]

While the first step of the procedure is obviously dominated by the above, this is not the case for the remaining ones.
In particular, building the schedule from the IP solution costs $\Oh((n+m)/\eps^2)$, if the procedure described in the proof of Lemma \ref{lem:split_IP_to_sched} is realized in a straight-forward fashion.
The last step of the algorithm is dominated by this, yielding the running time stated in the theorem below.
Note that the number of machines $m$ could be exponential in the number of jobs, and therefore the described procedure is a PTAS only for the special case of $m = \poly(n)$.
However, this limitation can be overcome with a little extra effort, as we discuss in Section \ref{sec:better_running_time}.
\begin{theorem}
The algorithm for the splittable model finds a schedule with makespan at most $(1 + 3\eps)T$ or correctly determines that there is no schedule with makespan $T$ in time $2^{\Oh(\nicefrac{1}{\eps^4}\log\nicefrac{1}{\eps})}n^2m\log m\log nm$.
\end{theorem}

\subsection{Preemptive Model}\label{sec:EPTAS_preemptive}

In the preemptive model we have to actually consider the time-line of the schedule on each machine instead of just the assignment of the jobs or job pieces, and this causes some difficulties.
For instance, we will have to argue that it suffices to look for a schedule with few possible starting points, and we will have to introduce additional constraints in the IP in order to ensure that pieces of the same job do not overlap.
Our first step, in dealing with this extra difficulty is to introduce some concepts and notation:
For a given schedule with a makespan bound $T$, we call a job piece together with its setup a \emph{block}, and we call the schedule $X$-layered, for some value $X$, if each block starts at a multiple of $X$.
Corresponding to this, we call the time in the schedule between two directly succeeding multiples of $X$ a \emph{layer} and the corresponding time on a single machine a \emph{slot}.
We number the layers bottom to top and identify them with their number, that is, the set of layers $\layers$ is given by $\sett{\ell\in\ZZ_{>0}}{(\ell-1) X\leq T}$.
Note that in an $X$-layered schedule, there is at most one block in each slot and for each layer there can be at most one block of each job present.
Furthermore, for $X$-layered schedules, we slightly alter the definition of free space:
We solely count the space from slots that are completely free.
If in such a schedule, for each job there is at most one slot occupied by this job but not fully filled, we additionally call the schedule \emph{layer-compliant}.

\subsubsection*{Simplification of the Instance}

In the preemptive model we distinguish \emph{big}, \emph{medium} and \emph{small} setup jobs, using two parameters $\delta$ and $\mu$:
The big setup jobs $\bsjobs$ are those with setup time at least $\delta T$, the small $\ssjobs$ have a setup time smaller than $\mu T$, and the medium $\msjobs$ are the ones in between.
We set $\mu = \eps^2\delta$ and we choose $\delta\in\set{\eps^1,\dots,\eps^{\nicefrac{2}{\eps^2}}}$ such that the summed up processing time together with the summed up setup time of the medium setup jobs is upper bounded by $m \eps T$, i.e., $\sum_{j\in\msjobs}(s_j + p_j)\leq m \eps T$.
If there is a schedule with makespan $T$, such a choice is possible, because of the pidgeon hole principle, and because the setup time of each job has to occur at least once in any schedule.
Similar arguments are widely used, e.g. in the context of geometrical packing algorithms.
Furthermore we distinguish the jobs by processing times, calling those with processing time at least $\eps T$ \emph{big} and the others \emph{small}.
For a given set of jobs $J$, we call the subsets of big or small jobs $J_\pbig$ or $J_\psmall$ respectively.
We perform three simplification steps, aiming for an instance in which the small and medium setup jobs are big; small setup jobs have setup time $0$; and for which an $\eps\delta T$-layered, layer-compliant schedule exists. 

Let $I_1$ be the instance we get by removing the small jobs with medium setup times $\msjobs_\psmall$ from the given instance $I$.
\begin{lemma}\label{lem_prmt_simplification_1}
If there is a schedule with makespan at most $T$ for $I$, there is also such a schedule for $I_1$, and if there is a schedule with makespan at most $T'$ for $I_1$ there is a schedule with makespan at most $T' + (\eps + \delta)T$ for $I$.
\end{lemma}
\begin{proof}
The first claim is obvious.
For the second, we create a sequence containing the jobs from $\msjobs_\psmall$ each directly preceded by its setup time.
Recall that the overall length of the objects in this sequence is at most $m\eps T$, and the length of each job is bounded by $\eps T$.
We greedily insert the objects from the sequence, considering each machine in turn.
On the current machine we start at time $T' + \delta T$ and keep inserting until $T' + \delta T + \eps T$ is reached.
If the current object is a setup time, we discard it and continue with the next machine and object.
If, on the other hand, it is a job, we split it, such that the remaining space on the current machine can be perfectly filled.
We can place all objects like this, however the first job part placed on a machine might be missing a setup.
We can insert the missing setups because they have length at most $\delta T$ and between time $T'$ and $T'+ \delta T$ there is free space.
\end{proof}
Next, we consider the jobs with small setup times:
Let $I_2$ be the instance we get by removing the small jobs with small setup times $\ssjobs_\psmall$ and setting the setup time of the big jobs with small setup times to zero, i.e., $\bar{s}_j=0$ for each $j\in \ssjobs_\pbig$.
Note that in the resulting instance each small job has a big setup time.
Furthermore, let $L := \sum_{j\in \ssjobs_\psmall} p_j + s_j$.
Then $L$ is an obvious lower bound for the space taken up by the jobs from $\ssjobs_\psmall$ in any schedule.
\begin{lemma}\label{lem_prmt_simplification_2}
If there is a schedule with makespan at most $T$ for $I_1$, there is also a $(T,L)$-schedule for $I_2$; and if there is a $\gamma T$-layered $(T',L)$-schedule for $I_2$, with $T'$ a multiple of $\gamma T$, there is also a schedule with makespan at most $(1 + \gamma^{-1}\mu)T'  + (\mu + \eps) T$ for $I_1$.
\end{lemma}
\begin{proof}
The first claim is obvious, and for the second consider a $\gamma T$-layered $(T',L)$-schedule for $I_2$.
We create a sequence that contains the jobs of $\ssjobs_\psmall$ and their setups, such that each job is directly preceded by its setup.
Remember that the remaining space in partly filled slots is not counted as free space.
Hence, since the overall length of the objects in the sequence is $L$, there is is enough space in the free slots of the schedule to place them.
We do so in a greedy fashion guaranteeing that each job is placed on exactly one machine:
We insert the objects from the sequence into the free slots, considering each machine in turn and starting on the current machine from the beginning of the schedule and moving on towards its end.
If an object cannot be fully placed into the current slot there are two cases:
It could be a job or a setup.
In the former case, we cut it and continue placing it in the next slot, or, if the current slot was the last one, we place the rest at the end of the schedule.
In the latter case, we discard the setup and continue with the next slot and object.
The resulting schedule is increased by at most $\eps T$, which is caused by the last job placed on a machine.

To get a proper schedule for $I_1$ we have to insert some setup times:
For the large jobs with small setup times and for the jobs that were cut in the greedy procedure.
We do so by inserting a time window of length $\mu T$ at each multiple of $\gamma T$ and at the end of the original schedule on each machine.
By this, the schedule is increased by at most $\gamma^{-1}\mu T' + \mu T$.
Since all the job parts in need of a setup are small and did start at multiples of $\mu T$ or at the end, we can insert the missing setups.
Note that blocks that span over multiple layers are cut by the inserted time windows.
This, however, can easily be repaired by moving the cut pieces properly down.
\end{proof}

We continue by rounding the medium and big setup and all the processing times.
In particular, we round the processing times and the big setup times up to the next multiple of $\eps\delta T$ and the medium setup times to the next multiple of $\eps\mu T$, i.e., $\bar{p}_j= \ceil{p_j/(\eps\delta T)} \eps\delta T$ for each job $j$, $\bar{s}_j= \ceil{s_j/(\eps\delta T)} \eps\delta T$ for each big setup job $j\in\bsjobs$, and $\bar{s}_j= \ceil{s_j/(\eps\mu T)} \eps\mu T$ for each medium setup job $j\in\msjobs_\pbig$.
\begin{lemma}\label{lem_prmt_simplification_3}
If there is a $(T,L)$-schedule for $I_2$, there is also an $\eps\delta T$-layered, layer-compliant $((1+3\eps)T,L)$-schedule for $I_3$; and if there is a $\gamma T$-layered $(T',L)$-schedule for $I_3$, there is also such a schedule for $I_2$.
\end{lemma}
While the second claim is easy to see, the proof of the first is rather elaborate and unfortunately a bit tedious.
Hence, since we believe Lemma \ref{lem_prmt_simplification_3} to be fairly plausible by itself, we postpone its proof to the end of the section and proceed discussing its use.

For the big and small setup jobs both processing and setup times are multiples of $\eps\delta T$.
Therefore, the length of each of their blocks in an $\eps\delta T$-layered, layer-compliant schedule is a multiple of $\eps\delta T$.
For a medium setup job, on the other hand, we know that the overall length of its blocks has the form $x\eps\delta T + y \eps\mu T$, with non-negative integers $x$ and $y$.
In particular it is a multiple of $\eps\mu T$, because $\eps\delta T = (1/\eps^2) \eps\mu T$.
In a $\eps\delta T$-layered, layer-compliant schedule, for each medium setup job the length of all but at most one block is a multiple of $\eps\delta T$ and therefore a multiple of $\eps\mu T$.
If both the overall length and the lengths of all but one block are multiples of $\eps\mu T$, this is also true for the one remaining block.
Hence, we will use the $\MCIP$ not to find an $\eps\delta T$-layered, layer-compliant schedule in particular, but an $\eps\delta T$-layered one with block sizes as described above and maximum free space.

Based on the simplification steps, we define two makespan bounds $\bar{T}$ and $\breve{T}$:
Let $\bar{T}$ be the makespan bound we get by the application of the Lemmata \ref{lem_prmt_simplification_1}-\ref{lem_prmt_simplification_3}, i.e., $\bar{T} = (1+3\eps)T$.
We will use the $\MCIP$ to find an $\eps\delta T$-layered $(\bar{T},L)$-schedule for $I_3$, and apply the Lemmata \ref{lem_prmt_simplification_1}-\ref{lem_prmt_simplification_3} backwards to get schedule for $I$ with makespan at most $\breve{T} = (1 + (\eps\delta)^{-1}\mu)\bar{T} + (\mu + \eps) T + (\eps + \delta)T \leq (1 + 9\eps)T$, using $\eps\leq 1/2$.

\subsubsection*{Utilization of the MCIP}

Similar to the splittable case, the basic objects are the (big) jobs, i.e.,  $\bobs = \jobs_{\pbig}$, and their single value is their processing time ($D = 1$).
The modules, on the other hand, are more complicated, because they additionally need to encode which layers are exactly used and, in case of the medium jobs, to which degree the last layer is filled.
For the latter we introduce buffers, representing the unused space in the last layer, and define modules as tuples $(\ell,q,s,b)$ of starting layer, job piece size, setup time and buffer size.
For a module $M = (\ell,q,s,b)$, we write $\ell_M = \ell$, $q_M=q$, $s_M=s$ and $b_M = b$, and we define the size $\size(M)$ of $M$ as $ s + q + b$.
The overall set of modules $\mods$ is the union of the modules for big, medium and small setup jobs $\bsmods$, $\msmods$ and $\ssmods$ that are defined in the following.
For this let 
$Q^{\mathrm{bst}} = \sett{q}{q = x\eps\delta T, x\in\ZZ_{> 0},q \leq \bar{T}}$ and
$Q^{\mathrm{mst}} = \sett{q}{q = x\eps\mu T, x\in\ZZ_{> 0},q \leq \bar{T}}$ be the sets of possible job piece sizes of big and medium setup jobs; 
$S^{\mathrm{bst}} = \sett{s}{s = x\eps\delta T, x\in\ZZ_{\geq 1/\eps},s \leq \bar{T}}$ and
$S^{\mathrm{mst}} = \sett{s}{s = x\eps\mu T, x\in\ZZ_{\geq 1/\eps},s \leq \delta T }$ be the sets of possible big and medium setup times; 
$B = \sett{b}{b = x\eps\mu T, x\in\ZZ_{\geq 0}, b < \eps\delta T}$ the set of possible buffer sizes; 
and $\Xi = \set{1,\dots,1/(\eps\delta) + 3/\delta}$ the set of layers. 
We set:
\begin{align*}
\bsmods &= \sett{(\ell,q,s,0)}{\ell\in \Xi,q\in Q^{\mathrm{bst}},s\in S^{\mathrm{bst}}, (\ell - 1)\eps\delta T + s + q\leq \bar{T} }\\
\msmods &= \sett{(\ell,q,s,b)\in \Xi\times Q^{\mathrm{mst}}\!\!\times S^{\mathrm{mst}}\!\!\times B}{x = s + q + b\in\eps\delta T\ZZ_{>0}, (\ell - 1)\eps\delta T + x\leq \bar{T} }\\
\ssmods &= \sett{(\ell,\eps\delta T,0,0)}{\ell\in \Xi}
\end{align*}
Concerning the small setup modules, note that the small setup jobs have a setup time of $0$ and therefore may be covered slot by slot.
We establish $\eps\mu T=1$ via scaling, to ensure integral values.
A big, medium or small job is eligible for a module, if it is also big, medium or small respectively and the setup times fit.

We have to avoid that two modules $M_1,M_2$, whose corresponding time intervals overlap, are used to cover the same job or in the same configuration.
Such an overlap is given, if there is some layer $\ell$ used by both of them, that is, $(\ell_M - 1)\eps\delta T \leq (\ell - 1)\eps\delta T < (\ell_M - 1)\eps\delta T + \size(M)$ for both $M \in \set{M_1,M_2}$.
Hence, for each layer $\ell\in\Xi$, we set $\mods_\ell\subseteq \mods$ to be the set of modules that use layer $\ell$.
Furthermore, we partition the modules into groups $\groups$ by size and starting layer, i.e., $\groups = \sett{G\subseteq \mods}{M,M'\in G\Rightarrow \size(M)=\size(M') \wedge \ell_M = \ell_{M'}}$.
The size of a group $G\in\groups$ is the size of a module from $G$, i.e. $\size(G) = \size(M)$ for $M\in G$.
Unlike before we consider \emph{configurations of module groups} rather than module sizes. 
More precisely, the set of configurations $\confs$ is given by the configurations of groups, such that for each layer at most one group using this layer is chosen, i.e., $\confs = \sett{C\in\ZZ_{\geq 0}^\groups}{\forall\ell\in\layers:\sum_{G\subseteq\mods_\ell}C_G\leq 1}$.
With this definition we prevent overlap conflicts on the machines.
Note that unlike in the cases considered so far, the size of a configuration does not correspond to a makespan in the schedule, but to used space, and the makespan bound is realized in the definition of the modules instead of in the definition of the configurations.
To also avoid conflicts for the jobs, we extend the basic $\MCIP$ with additional locally uniform constraints.
In particular, the constraints of the extended $\MCIP$ for the above definitions with adapted notation and without duplication of the configuration variables are given by:
\begin{align}
\sum_{C\in\confs}x_{C} & = m & \label{eq:MCIP_machs_preempt} \\ 
\sum_{C\in\confs(T)} C_G x_{C} & = \sum_{j\in\jobs}\sum_{M\in G}y^{(j)}_M & \forall G \in\groups \label{eq:MCIP_mods_preempt}\\
\sum_{M\in\mods}q_M y^{(j)}_M & = p_j & \forall j\in\jobs \label{eq:MCIP_jobs_preempt}\\
\sum_{M\in\mods_\ell} y^{(j)}_M & \leq 1 & \forall j\in\jobs, \ell\in\Xi \label{eq:MCIP_colis_preempt}
\end{align}
Like in the first two cases we minimize the summed up size of the configurations, via the objective function $\sum_C \size(C)x_C$.
Note that in this case the size of a configuration does not have to equal its height.
It is easy to see that the last constraint is indeed locally uniform.
However, since we have an inequality instead of an equality, we have to introduce $|\layers|$ slack variables in each brick, yielding:
\begin{observation}\label{rem:MCIP_prmt}
The $\MCIP$ extended like above is an $n$-fold IP with brick-size $t = |\mods| + |\confs| + |\Xi|$, brick number $n = |\jobs|$, $r = |\groups|+1$ globally uniform and $s = D + |\Xi|$ locally uniform constraints.
\end{observation}
\begin{lemma}\label{lem:prmt_IP_to_sched}
With the above definitions, there is an $\eps\delta T$-layered ($\bar{T}, L$)-schedule for $I_3$ in which the length of a block is a multiple of $\eps\delta T$, if it belongs to a small or big setup job, or a multiple of $\eps\mu T$ otherwise, iff the extended $\MCIP$ has a solution with objective value at most $m\bar{T} - L$.
\end{lemma}
\begin{proof}
We first consider such a schedule for $I_3$.
For each machine, we can derive a configuration that is given by the starting layers of the blocks together with the summed up length of the slots the respective block is scheduled in.
The size of the configuration $C$ is equal to the used space on the respective machine.
Hence, we can fix some arbitrary job $j$ and set $x^{(j)}_C$ to the number of machines corresponding to $j$ (and $x^{(j')}_C=0$ for $j'\neq j$).
Keeping in mind that in an $\eps\delta T$-layered schedule the free space is given by the free slots, the above definition yields an objective value bounded by $m\bar{T} - L$, because there was free space of at least $L$.
Next, we consider the module variables for each job $j$ in turn:
If $j$ is a small setup job, we set $y^{(j)}_{(\ell,\eps\delta T,0,0)}$ to $1$, if the $j$ occurs in $\ell$, and to $0$ otherwise.
Now, let $j$ be a big setup job.
For each of its blocks, we set $y^{(j)}_{(\ell,z-s_j,s_j,0)} = 1$, where $\ell$ is the starting layer and $z$ the length of the block.
The remaining variables are set to $0$.
Lastly, let $j$ be a medium setup job.
For each of its blocks, we set $y^{(j)}_{(\ell,z-s_j,s_j,b)} = 1$, where $\ell$ is the starting layer of the block, $z$ its length and $b = \ceil{z/(\eps\delta T)}\eps\delta T - z$.
Again, the remaining variables are set to $0$.
It is easy to verify that all constraints are satisfied by this solution.

If, on the other hand, we have a solution $(x,y)$ to the $\MCIP$ with objective value at most $m\bar{T} - L$, we reserve $\sum_{j}x^{(j)}_{C}$ machines for each configuration $C$.
There are enough machines to do this, because of (\ref{eq:MCIP_machs_preempt}).
On each of these machines we reserve space:
For each $G\in\groups$, we create an allocated space of length $\size(G)$ starting from the starting layer of $G$, if $C_G = 1$.
Let $j$ be a job and $\ell$ be a layer.
If $j$ has a small setup time, we create $y^{(j)}_{(\ell,\eps\delta T,0,0)} $ pieces of length $\eps\delta T$ and place these pieces into allocated spaces of length $\eps\delta T$ in layer $\ell$.
If, on the other hand, $j$ is a big or medium setup job, we consider each possible job part length $q\in Q^{\mathrm{bst}}$ or $q\in Q^{\mathrm{mst}}$, create $y^{(j)}_{(\ell,q,s_j,0)} $ or $y^{(j)}_{(\ell,q,s_j,b)}$, with $b = \ceil{q/(\eps\delta T)}\eps\delta T - \eps\delta T$, pieces of length $q$, and place them together with their setup time into allocated spaces of length $q$ in layer $\ell$.
Because of (\ref{eq:MCIP_jobs_preempt}) the entire job is split up by this, and because of (\ref{eq:MCIP_mods_preempt}) there are enough allocated spaces for all the job pieces.
The makespan bound is ensured by the definition of the modules, and overlaps are avoided, due to the definition of the configurations and (\ref{eq:MCIP_colis_preempt}).
Furthermore, the used slots have an overall length equal to the objective value of $(x,y)$ and therefore there is at least $L$ free space. 
\end{proof}

\subsubsection*{Result}

Summing up the above considerations, we get:
\begin{algorithm}\
	\begin{enumerate}
   	\item If there is no suitable class of medium setup jobs, report that there is no schedule with makespan $T$ and terminate the procedure.
		\item Generate the modified instance $I_3$: 
		\begin{itemize}
			\item Remove the small jobs with medium setup times.
			\item Remove the small jobs with small setup times, and decrease the setup time of big jobs with small setup time to $0$.
			\item Round the big processing times, as well as the medium, and the big setup times.
		\end{itemize}
		\item Build and solve the $\MCIP$ for $I_3$. 
		\item If the $\MCIP$ is infeasible, or the objective value greater than $m \bar{T} - L$, report that $I$ has no solution with makespan $T$.
		\item Otherwise build the $\eps\delta T$-layered schedule with makespan $\bar{T}$ and free space at least $L$ for $I_3$.
		\item Transform the schedule into a schedule for $I$ with makespan at most $\breve{T}$:
		\begin{itemize}
			\item Use the prerounding processing and setup times. 
			\item Insert the small jobs with small setup times into the free slots and insert the setup times of the big jobs with small setup times. 
			\item Insert the small jobs with medium setup times.
		\end{itemize}
	\end{enumerate}
\end{algorithm}
We analyze the running time of the procedure, and start by bounding the parameters of the extended $\MCIP$.
We have $|\bobs| = n$ and $D=1$ by definition, and the number of layers $|\layers|$ is obviously $\Oh(1/(\eps\delta)) = \Oh(1/\eps^{2/\eps + 1}) = 2^{\Oh(\nicefrac{1}{\eps}\log\nicefrac{1}{\eps})}$.
Furthermore, it is easy to see that $|Q^{\mathrm{bst}}| = \Oh(1/(\eps\delta))$, $|Q^{\mathrm{mst}}| = \Oh(1/(\eps^3\delta))$, $|S^{\mathrm{bst}}| = \Oh(1/(\eps\delta))$, $|S^{\mathrm{mst}}| = \Oh(1/(\eps^3))$, and $|B|=\Oh(1/\eps^2)$. 
This gives us $\bsmods \leq |\layers||Q^{\mathrm{bst}}||S^{\mathrm{bst}}|$, $\msmods \leq |\layers||Q^{\mathrm{mst}}||S^{\mathrm{mst}}||B| $ and $\ssmods = |\layers|$, and therefore $|\mods| = |\bsmods| + |\msmods| + |\ssmods| = 2^{\Oh(\nicefrac{1}{\eps}\log\nicefrac{1}{\eps})}$.
Since their are $\Oh(1/(\delta\eps)) $ distinct module sizes, the number of groups $|\groups|$ can be bounded by $\Oh(|\layers|/(\eps\delta)) = 2^{\Oh(\nicefrac{1}{\eps}\log\nicefrac{1}{\eps})}$.
Hence, for the number of configurations we get $|\confs| = \Oh((1/(\eps\delta))^{|\groups|}) = 2^{2^{\Oh(\nicefrac{1}{\eps}\log\nicefrac{1}{\eps})}}$.
By Observation \ref{rem:MCIP_prmt}, the modified $\MCIP$ has $r = 2^{\Oh(\nicefrac{1}{\eps}\log\nicefrac{1}{\eps})}$ many globally and $s = 2^{\Oh(\nicefrac{1}{\eps}\log\nicefrac{1}{\eps})}$ many locally uniform constraints; its brick number is $n$, and its brick size is $ t = 2^{2^{\Oh(\nicefrac{1}{\eps}\log\nicefrac{1}{\eps})}}$.
All occurring values in the matrix are bounded by $\bar{T}$, yielding $\Delta \leq \bar{T} = 1/(\eps\mu) + 1/\mu =  2^{\Oh(\nicefrac{1}{\eps}\log\nicefrac{1}{\eps})}$, due to the scaling step.
Furthermore, the numbers in the input can be bounded by $m2^{\Oh(\nicefrac{1}{\eps}\log\nicefrac{1}{\eps})}$ and all variables can be upper bounded by $\Oh(m)$.
Hence, we have $\varphi = \Oh(\log m + 1/\eps \log 1/\eps)$ and $\Phi = \Oh(m)$, and due to Theorem \ref{thm:solving_n-fold} we can solve the $\MCIP$ in time:
\[(rs\Delta)^{\Oh(r^2s+rs^2)}t^2n^2\varphi\log(\Phi)\log(nt\Phi) = 2^{2^{\Oh(\nicefrac{1}{\eps}\log\nicefrac{1}{\eps})}}n^2\log^2 m\log nm\]
A straight-forward realization of the procedure for the creation of the $\eps\delta T$-layered $(\bar{T},L)$-schedule for $I_3$ (the fifth step), which is described in the proof of Lemma \ref{lem:prmt_IP_to_sched}, will take $nm2^{\Oh(\nicefrac{1}{\eps}\log\nicefrac{1}{\eps})}$ time, yielding:
\begin{theorem}
The algorithm for the preemptive model finds a schedule with makespan at most $(1+9\eps)T$ or correctly determines that there is no schedule with makespan $T$ in time $2^{2^{\Oh(\nicefrac{1}{\eps}\log\nicefrac{1}{\eps})}}n^2 m\log m\log nm$.
\end{theorem}

\subsubsection*{Proof of Lemma \ref{lem_prmt_simplification_3}}

We divide the proof into three steps, which can be summarized as follows:
\begin{enumerate}
\item We transform a $(T,L)$-schedule for $I_2$ into a $((1+3\eps)T,L)$-schedule for $I_3$ in which the big setup jobs are already properly placed inside the layers.
\item We construct a flow network with integer capacities and a maximum flow, based on the placement of the remaining jobs in the layers.
\item Using flow integrality and careful repacking, we transform the schedule into a $\eps\delta T$-layered, layer-compliant schedule.
\end{enumerate}
More precisely the above transformation steps will produce a $\eps\delta T$-layered, layer-compliant $((1+3\eps)T,L)$-schedule with the additional properties, that too much processing time may be inserted for some jobs or setup times are produced that are not followed by the corresponding job pieces.
Note that this does not cause any problems:
We can simply remove the extra setups and processing time pieces.
For the medium jobs this results in a placement with at most one used slot that is not fully filled, as required in a layer-compliant schedule.

\subparagraph*{Step 1.}

Remember that a block is a job piece together with its setup time placed in a given schedule.
Consider a $(T,L)$-schedule for $I_2$ and suppose that for each block in the schedule there is a container perfectly encompassing it.
Now, we stretch the entire schedule by a factor of $(1+3\eps)$ and in this process we stretch and move the containers correspondingly.
The blocks are not stretched but moved in order to stay in their container, and we assume that they are positioned at the bottom, that is, at the beginning of the container.
Note that we could move each block inside its respective container without creating conflicts with other blocks belonging to the same job. 
In the following, we use the extra space to modify the schedule.
Similar techniques are widely used in the context of geometric packing algorithms. 
	
Let $j$ be a big setup job. 
In each container containing a block belonging to $j$, there is a free space of at least $3\eps\delta T$, because the setup time of $j$ is at least $\delta T$ and therefore the container had at least that length before the stretching.
Hence, we have enough space to perform the following two steps.
We move the block up by at most $\eps\delta T$, such that it starts at a multiple of $\eps\delta T$.
Next, we enlarge the setup time and the processing time by at most $\eps\delta T$, such that both are multiples of $\eps\delta T$.	
Now the setup time is equal to the rounded setup time, while the processing time might be bigger, because we performed this step for each piece of the job.
We outline the procedure in Figure \ref{fig:container}.
\begin{figure}
	\centering
	\begin{tikzpicture}[scale=0.8]
	\tikzmath{ 
   	\nl1 = 13; 
		\l1 = 0.35; 
		\f1 = 11/8;
		\x1 = 0;
		\x2 = 13;
		\xa1 = 1;
		\xa2 = 3;
		\g1 = 0.4; 
		\s1 = 4.12* 2; 
		\s2 = 9; 
		\p1 = 0.3* 2; 
		\p2 = 1; 
		\yai0 = \g1; 
		\yai1 = \g1; 
		\yai2 = \yai1 + \s1; 
		\yai3 = \yai2 + \p1; 
		\yai4 = \yai3; 
		\xb1 = 4;
		\xb2 = 6;
		\ybi0 = \g1 * \f1; 
		\ybi1 = \g1 * \f1; 
		\ybi2 = \ybi1 + \s1; 
		\ybi3 = \ybi2 + \p1; 
		\ybi4 = \ybi0 + \f1*(\s1 + \p1); 
		\xc1 = 7;
		\xc2 = 9;
		\yci0 = \ybi0; 
		\yci1 = 1; 
		\yci2 = \yci1 + \s1; 
		\yci3 = \yci2 + \p1; 
		\yci4 = \ybi4; 
		\xd1 = 10;
		\xd2 = 12;
		\ydi0 = \ybi0; 
		\ydi1 = 1; 
		\ydi2 = \ydi1 + \s2; 
		\ydi3 = \ydi2 + \p2; 
		\ydi4 = \ybi4; 
}

	\draw[thick,fill=lightgray] (\xa1,{\yai1*\l1}) rectangle (\xa2,{\yai2*\l1});
	\draw[thick,pattern = north east lines] (\xa1,{\yai1*\l1}) rectangle (\xa2,{\yai2*\l1});
	\draw[thick,fill=lightgray] (\xa1,{\yai2*\l1}) rectangle (\xa2,{\yai3*\l1});
	\draw[very thick] (\xa1,{\yai0*\l1}) rectangle (\xa2,{\yai4*\l1});
	
   
	\draw[thick,fill=lightgray] (\xb1,{\ybi1*\l1}) rectangle (\xb2,{\ybi2*\l1});
	\draw[thick,pattern = north east lines] (\xb1,{\ybi1*\l1}) rectangle (\xb2,{\ybi2*\l1});
	\draw[thick,fill=lightgray] (\xb1,{\ybi2*\l1}) rectangle (\xb2,{\ybi3*\l1});
	\draw[very thick] (\xb1,{\ybi0*\l1}) rectangle (\xb2,{\ybi4*\l1});
   
	\draw[thick,fill=lightgray] (\xc1,{\yci1*\l1}) rectangle (\xc2,{\yci2*\l1});
	\draw[thick,pattern = north east lines] (\xc1,{\yci1*\l1}) rectangle (\xc2,{\yci2*\l1});
	\draw[thick,fill=lightgray] (\xc1,{\yci2*\l1}) rectangle (\xc2,{\yci3*\l1});
	\draw[very thick] (\xc1,{\yci0*\l1}) rectangle (\xc2,{\yci4*\l1});
	
	\draw[thick,fill=lightgray] (\xd1,{\ydi1*\l1}) rectangle (\xd2,{\ydi2*\l1});
	\draw[thick,pattern = north east lines] (\xd1,{\ydi1*\l1}) rectangle (\xd2,{\ydi2*\l1});
	\draw[thick,fill=lightgray] (\xd1,{\ydi2*\l1}) rectangle (\xd2,{\ydi3*\l1});
   
	\foreach \i in {0,...,\nl1}
	{	\draw[dashed] (\x1,\i * \l1) -- (\x2,\i * \l1);
	}
	
	\end{tikzpicture}
	\caption{The stretching and rounding steps, for a small job part with big setup time starting in the first layer of the schedule, depicted from left to right: The schedule and the containers are stretched; the block is moved up; and the processing and the setup time are increased. The hatched part represents the setup time, the thick rectangle the container, and the dashed lines the layers, with $\eps = \delta = 1/8$.}
	\label{fig:container}
\end{figure}
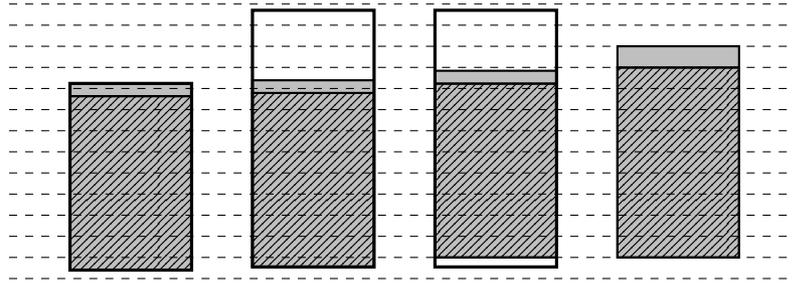
	
We continue with the small setup jobs.
These jobs are big and therefore for each of them there is a summed up free space of at least $3\eps^2 T$ in the containers belonging to the respective job---more than enough to enlarge some of the pieces such that their overall length matches the rounded processing time.  
	
Lastly, we consider the medium setup jobs.
These jobs are big as well and we could apply the same argument as above, however, we need to be a little bit more careful in order to additionally realize the rounding of the setup times and an additional technical step, we need in the following.
Fix a medium setup job $j$ and a container filled with a block belonging to $j$.
Since the setup time has a length of at least $\mu T$, the part of the container filled with it was increased by at least $3\eps\mu T$.
Hence, we can enlarge the setup time to the rounded setup time without using up space in the container that was created due to the processing time part.
We do this for all blocks belonging to medium setup jobs.
The extra space in the containers of a medium setup job due to the processing time parts is still at least $3\eps^2 T\geq 3\eps\delta T$.
For each medium setup job $j$ we spend at most $\eps\delta T$ of this space to enlarge its processing time to its rounded size and again at most $\eps\delta T$ to create a little bit of extra processing time in the containers belonging to $j$.
The size of this extra processing time is bounded by $\eps\delta T$ and chosen in such a way that the overall length of all blocks belonging to $j$ in the schedule is also a multiple of $\eps\delta T$.
Because of the rounding, the length of the added extra processing time for each $j$ is a multiple of $\eps\mu T$.
The purpose of the extra processing time is to ensure integrality in the flow network, which is constructed in the next step. 
	
Note that the free space that was available in the original schedule was not used in the above steps, in fact it was even increased by the stretching.
Hence, we have created a $((1+3\eps)T,L)$-schedule for $I_3$---or a slightly modified version thereof---and the big setup jobs are already well behaved with respect to the $\eps\delta T$-layers, that is, they start at multiples of $\eps\delta T$, and fully fill the slots they are scheduled in.

\subparagraph*{Step 2.}

Note that for each job $j$ and layer $\ell\in\layers$, the overall length $q_{j,\ell}$ of job and setup pieces belonging to $j$ and placed in $\ell$ is bounded by $\eps\delta T$.
We say that $j$ is \emph{fully}, or \emph{partially}, or \emph{not} scheduled in layer $\ell$, if $q_{j,\ell} = 1$, or $q_{j,\ell} \in (0,1)$, or $q_{j,\ell} = 0$ respectively.
Let $X_j$ be the set of layers in which $j$ is scheduled partially and $Y_\ell$ the set of (medium or small setup) jobs partially scheduled in~$\ell$.
Then $a_j = \sum_{\ell\in X_j} q_{j,\ell}$ is a multiple of $\eps\delta T$ and we set $n_j = a_j/(\eps\delta T)$.
Furthermore, let $b_\ell = \sum_{j\in Y_\ell} q_{j,\ell} $ and $k_\ell = \ceil{b_\ell/(\eps\delta T)}$.
	
Our flow network has the following structure: 
There is a node $v_j$ for each medium or small setup job, and a node $u_\ell$ for each layer $\ell$, as well as a source $\alpha$ and a sink $\omega$.
The source node is connected to the job nodes via edges $(\alpha,v_j)$ with capacity $n_j$; and the layer nodes are connected to the sink via edges $(u_\ell,\omega)$ with capacity $k_\ell$.
Lastly, there are edges $(v_j,u_\ell)$ between job and layer nodes with capacity $1$, if $j$ is partially scheduled in layer $\ell$, or $0$ otherwise.
In Figure \ref{fig:Network} a sketch of the network is given.
	
The schedule can be used to define a flow $f$ with value $\sum_{j}n_j$ in the network, by setting $f(\alpha, v_j) = n_j$, $f(u_\ell, \omega) = b_\ell / (\eps\delta T)$, and $f(v_j,u_\ell) = q_{j,\ell}/(\eps\delta T)$.
It is easy to verify that $f$ is a maximum flow, and because all capacities in the flow network are integral, we can find another maximum flow $f'$ with integral values.
\begin{figure}	
\centering
\def\layersep{2.5cm}
\begin{tikzpicture}
[shorten >=1pt, node distance=\layersep]
    \tikzstyle{every pin edge}=[shorten <=3pt]
    \tikzstyle{neuron}=[circle,draw = black,minimum size=20pt]
    \tikzstyle{annot} = [text width=4em, text centered]
        
    \node[neuron] (I-1) at (0,4) {};
    \node (I-2) at (0,3) {$\vdots$};
    \node[neuron] (I-3) at (0,2) {$v_j$};
    \node (I-4) at (0,1) {$\vdots$};
    \node[neuron] (I-5) at (0,0) {};    
        
    \node[neuron, left of=I-3] (S) {$\alpha$};    
    
	\path (S) edge node[above]{$n_j$} (I-3);
	\path (S) edge (I-1);
	\path (S) edge (I-5);

    \node[neuron, right of=I-1] (H-1) {};
    \node[right of=I-2] (H-2) {$\vdots$};
    \node[neuron, right of=I-3] (H-3) {$u_\ell$};
    \node[right of=I-4] (H-4) {$\vdots$};
    \node[neuron, right of=I-5] (H-5) {};
    
    \path (I-1) edge node[above]{$1/0$} (H-1);
    \path (I-1) edge (H-3);
    \path (I-1) edge (H-5);
	\path (I-3) edge (H-1);
	\path (I-3) edge (H-3);
	\path (I-3) edge (H-5);
	\path (I-5) edge (H-1);
	\path (I-5) edge (H-3);
	\path (I-5) edge (H-5);

    \node[neuron, right of=H-3] (T) {$\omega$}; 
    
	\path (H-1) edge (T);
	\path (H-3) edge node[above]{$k_\ell$} (T);
	\path (H-5) edge (T); 

%

   \node[annot,above of=I-1, node distance=1cm] (hj) {Jobs};
    \node[annot,right of=hj] (hl) {Layer};

\end{tikzpicture}
\caption{Flow network for layers and partially scheduled jobs.}
\label{fig:Network}
\end{figure}
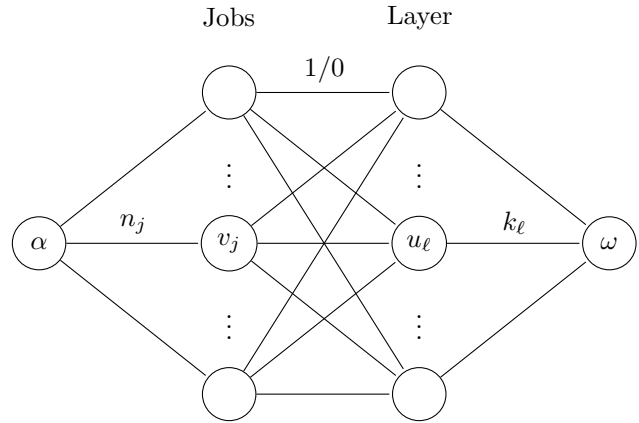

\subparagraph*{Step 3.}

We start by introducing some notation and a basic operation for the transformation of the schedule:
Given two machines $i$ and $i'$ and a time $t$, a \emph{machine swap} between $i$ and $i'$ at moment $t$ produces a schedule, in which everything that was scheduled on $i$ from $t$ on is now scheduled on $i'$ and vice versa.
If on both machines there is either nothing scheduled at $t$, or blocks are starting or ending at $t$, the resulting schedule is still feasible.
Moreover, if there is a block starting at $t$ on one of the machines and another one belonging to the same job ending on the other we can merge the two blocks and transform the setup time of the first into processing time.
We assume in the following that we always merge if this is possible, when performing a machine swap.
Remember that by definition blocks belonging to the same job cannot overlap.
However, if there was overlap, it could be eliminated using machine swaps \cite{schuurman1999preemptive}.

If a given slot only contains pieces of jobs that are partially scheduled in the layer, we call the slot \emph{usable}.
Furthermore, we say that a job $j$ is \emph{flow assigned} to layer $\ell$, if $f'(v_j,u_\ell) = 1$.
In the following, we will iterate through the layers, and create as many usable slots as possible, reserve them for flow assigned jobs, and fill them with processing and setup time of the corresponding slot later on.
To do so, we have to distinguish different types of blocks belonging to jobs that are partially placed in a given layer:
Inner blocks, which lie completely inside the layer and touch at most one of its borders; upper cross-over blocks, which start inside the layer and end above it; and lower cross-over blocks, which start below the layer and end inside it.
When manipulating the schedule layer by layer, the cross-over jobs obviously can cause problems.
To deal with this, we will need additional concepts:
A \emph{repair piece} for a given block is a piece of setup time of length less than $\eps\delta T$, with the property that the block and the repair piece together make up exactly one setup of the respective job. 
Hence, if a repair-piece is given for a block, the block is comprised completely of setup time.
Moreover, we say that a slot reserved for a job $j$, has a \emph{dedicated setup}, if there is a block of $j$ including a full setup \emph{starting or ending} inside the slot.

In the following, we give a detailed description of the transformation procedure followed by a high-level overview of the procedure.
The procedure runs through two phases. 
In the first phase the layers are transformed one after another from bottom to top.
After a layer is transformed the following invariants will always hold:
\begin{enumerate}
\item A scheduled block either includes a full setup, or has a repair piece, and in the latter case it was an upper cross-over block in a previous iteration.
\item Reserved slots that are not full have a dedicated setup.
\end{enumerate}
Note that the invariants are trivially fulfilled in the beginning.
During the first phase, we remove some job and setup parts from the schedule that are reinserted into the reserved slots in the second phase.
Let $\ell\in\layers$ denote the current layer.

In the first step, our goal is to ensure that jobs that are fully scheduled in $\ell$ occupy exactly one slot, thereby creating as many usable slots as possible.
Let $j$ be a job that is fully scheduled in layer $\ell$.
If there is a block belonging to $j$ and ending inside the layer at time $t$, there is another block belonging to $j$ and starting at $t$, because $j$ is fully scheduled in $\ell$ and there are no overlaps.
Hence, we can perform a machine swap at time $t$ between the two machines the blocks are scheduled on.
We do so, for each job fully scheduled in the layer and each corresponding pair of blocks.
After this step, there are at least $k_\ell$ usable slots and at most $k_\ell$ flow assigned jobs in layer $\ell$.

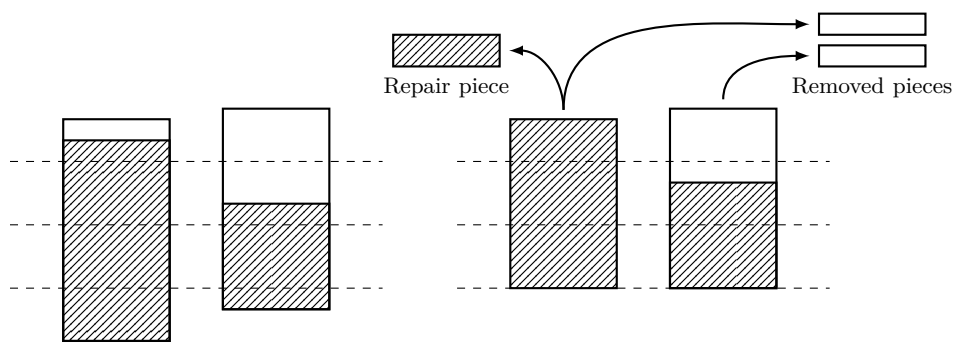
\begin{figure}
\centering
\begin{tikzpicture}[scale=0.14] 
\tikzmath{ 
\lsep  = 6;
\psep = 7; 
\bwidth = 10;
\bsep = 5;
\bborder = 5;  	
\lwidth = \bborder + \bwidth + \bsep + \bwidth + \bborder;
\rAbot = 5;
\rBbot = 2;
\rAtop = 4;
\rBtop = 5;
\rAtopS = 2;
\rBtopS = 2;
\rAleftbotX1 = \bborder;
\rAleftbotY1 = \lsep - \rAbot;
\rAlrighttopX1 = \bborder + \bwidth;
\rArighttopY1 = 3 * \lsep + \rAtop;
\rAlrighttopSX1 = \bborder + \bwidth;
\rArighttopSY1 = 3 * \lsep + \rAtopS;
\rBleftbotX1 = \bborder + \bwidth + \bsep;
\rBleftbotY1 = \lsep - \rBbot;
\rBlrighttopX1 = \bborder + \bwidth + \bsep + \bwidth;
\rBrighttopY1 = 3 * \lsep + \rBtop;
\rBrighttopSX1 = \bborder + \bwidth + \bsep + \bwidth;
\rBrighttopSY1 = 2 * \lsep + \rBtopS;
\rAleftbotX2 = \lwidth + \psep + \bborder;
\rAleftbotY2 = \lsep;
\rAlrighttopX2 = \lwidth + \psep + \bborder + \bwidth;
\rArighttopY2 = 3 * \lsep + \rAtop;
\rBleftbotX2 = \lwidth + \psep + \bborder + \bwidth + \bsep;
\rBleftbotY2 = \lsep;
\rBlrighttopX2 = \lwidth + \psep + \bborder + \bwidth + \bsep + \bwidth;
\rBrighttopY2 = 3 * \lsep + \rBtop;
\rBlrighttopSX2 = \lwidth + \psep + \bborder + \bwidth + \bsep + \bwidth;
\rBrighttopSY2 = 2 * \lsep + \rBtopS + \rBbot;
\repairleftbotX = \lwidth + 1;
\repairleftbotY = 4*\lsep + 3;
\repairrighttopX = \lwidth + 1 + \bwidth;
\repairrighttopY = 4*\lsep + 3 + \rAbot - \rAtopS;
\restlbX1 = 2 * \lwidth + \psep - 1;
\restlbY1 = 4*\lsep + 3;
\resttrX1 = 2 * \lwidth + \psep - 1 + \bwidth;
\resttrY1 = 4*\lsep + 3 + \rBbot;
\restlbX2 = 2 * \lwidth + \psep - 1;
\restlbY2 = 4*\lsep + 3 + \rBbot + 1;
\resttrX2 = 2 * \lwidth + \psep - 1 + \bwidth;
\resttrY2 = 4*\lsep + 3 + \rBbot + 1 + \rAtopS;
}

\foreach \i in {1,...,3}
{	
\draw[dashed] (0,\i * \lsep) -- (\lwidth,\i * \lsep);
\draw[dashed] ({(\lwidth + \psep)},\i * \lsep) -- ({(2 * \lwidth + \psep)},\i * \lsep);
}

\draw[thick] (\rAleftbotX1,\rAleftbotY1) rectangle (\rAlrighttopX1,\rArighttopY1);
\draw[thick,pattern = north east lines] (\rAleftbotX1,\rAleftbotY1) rectangle (\rAlrighttopSX1,\rArighttopSY1);
\draw[thick] (\rBleftbotX1,\rBleftbotY1) rectangle (\rBlrighttopX1,\rBrighttopY1);
\draw[thick,pattern = north east lines] (\rBleftbotX1,\rBleftbotY1) rectangle (\rBrighttopSX1,\rBrighttopSY1);

\draw[thick,pattern = north east lines] (\rAleftbotX2,\rAleftbotY2) rectangle (\rAlrighttopX2,\rArighttopY2);
\draw[thick] (\rBleftbotX2,\rBleftbotY2) rectangle (\rBlrighttopX2,\rBrighttopY2);
\draw[thick,pattern = north east lines] (\rBleftbotX2,\rBleftbotY2) rectangle (\rBlrighttopSX2,\rBrighttopSY2);

\draw[thick,pattern = north east lines] (\repairleftbotX,\repairleftbotY) rectangle (\repairrighttopX,\repairrighttopY);
\draw[thick] (\restlbX1, \restlbY1) rectangle (\resttrX1, \resttrY1);
\draw[thick] (\restlbX2, \restlbY2) rectangle (\resttrX2,\resttrY2);

\node[font=\footnotesize] at (\repairleftbotX + \bwidth/2,\repairleftbotY-2) {Repair piece};
\node[font=\footnotesize] at (\restlbX1 + \bwidth/2, \restlbY1 -2) {Removed pieces};

\node (A1) at (\rAlrighttopX2- \bwidth/2,\rArighttopY2 ) {};
\node (A2) at (\rBlrighttopX2- \bwidth/2,\rBrighttopY2 ) {};
\node (R) at (\repairrighttopX,{\repairrighttopY - (\rAbot - \rAtopS)/2}) {};
\node (P1) at (\restlbX1, \restlbY1 + 1) {};
\node (P2) at (\restlbX2, \restlbY2 + 1) {};

\draw [->, >=latex, thick] (A1) to [out= 90 ,in=0] (R);
\draw [->, >=latex, thick] (A1) to [out= 90 ,in=180] (P2);
\draw [->, >=latex, thick] (A2) to [out= 90 ,in=180] (P1);

\end{tikzpicture}
\caption{The rectangles represent blocks, the hatched parts the setup times, and the dashed lines layer borders. The push and cut step is performed on two blocks. For one of the two a repair piece is created.}
\label{fig:pushandcut}
\end{figure}
Next, we consider upper cross-over blocks of jobs that are partially scheduled in the layer $\ell$ but are not flow assigned to it.
These are the blocks that cause the most problems, and we perform a so-called \emph{push and cut step} (see Figure \ref{fig:pushandcut}) for each of them:
If $q$ is the length of the part of the block lying in $\ell$, we cut away the upper part of the block of length $q$ and move the remainder up by $q$.
If the piece we cut away does contain some setup time, we create a repair piece for the block out of this setup time.
The processing time part of the piece, on the other hand, is removed.
Note that this step preserves the first invariant.
The repair piece is needed in the case that the job corresponding to the respective block is flow assigned to the layer in which the block ends.

We now remove all inner blocks from the layer, as well as the parts of the upper and lower cross-over blocks that lie in the layer.
After this, all usable slots are completely free.
Furthermore, note that the the first invariant might be breached by this.

Next, we arbitrarily reserve usable slots for jobs flow assigned to the layer. 
For this, note that due to the definition of the flow network, there are at most $k_\ell$ jobs flow assigned to the layer and there are at least as many usable slots, as noted above.

Using machine swaps at the upper and lower border of the layer, we then ensure, that the upper and lower cross-over blocks of the jobs flow assigned to the layer lie on the same machine as the reserved slot.
This step might breach the second invariant as well.

However, for each job $j$ flow assigned to the layer, we perform the repair steps in order to restore the invariants:
If there is an upper cross-over block for $j$, we reinsert the removed part of the block at the end of the slot, thereby providing a dedicated setup for the remaining free space in the slot.
If there is a lower, but no upper cross-over block for $j$, there are two cases: 
Either there was a repair piece for the block or not.
In both cases we reinsert the removed part of the block in the beginning of the slot and in the first we additionally insert as much setup of the repair piece as possible.
The possible remainder of the repair piece is removed.
Now the slot is either full, or a full setup is provided.
If there is neither an upper nor a lower block for $j$, there is an in inner block belonging to $j$.
This has to be the case, because otherwise the capacity in the flow network between $j$ and $\ell$ is $0$ and $j$ could not have been flow assigned to $\ell$.
Moreover, this inner block contains a full setup and we can place it in the beginning of the slot, thus providing the dedicated setup.
The invariants are both restored.

After the first phase is finished, we have to deal with the removed pieces in the second one.
The overall length of the reserved slots for a job $j$ equals the overall length $a_j$ of its setup and job pieces from layers in which $j$ was partially scheduled.
Since, we did not create or destroy any job piece, we can place the removed pieces corresponding to job $j$ perfectly into the remaining free space of the slots reserved for $j$, and we do so after transforming them completely into processing time.
Because of the second invariant, there is a dedicated setup in each slot, however, it may be positioned directly above the newly inserted processing time.
This can be fixed by switching the processing time with the top part of the respective setup time. 

Lastly, all remaining usable slots are completely free at the end of this procedure, and since the others are full they have an overall size of at least $L$. 
We conclude the proof of Lemma \ref{lem_prmt_simplification_3} with an overview of the transformation procedure.
\begin{algorithm}
\ 

\noindent \emph{Phase 1:} For each layer $\ell\in\layers$, considered bottom to top, perform the following steps:
\begin{enumerate}
   \item Use machine swaps to ensure that jobs fully scheduled in $\ell$ occupy exactly one slot.
   \item For each upper cross-over block of a job partially scheduled but not flow assigned to $\ell$ perform a push and cut step.
   \item Remove inner blocks and parts of cross-over blocks that lie in $\ell$.
   \item Reserve usable slots for jobs flow assigned to the layer.
   \item Use machine swaps to ensure, that cross-over blocks of flow assigned jobs lie on the same machine as the reserved slot.
   \item For each job $j$ flow assigned to the layer, perform exactly one of the repair steps.
\end{enumerate}
\noindent \emph{Phase 2:} 
\begin{enumerate}
   \item Transform all removed pieces into processing time and insert the removed pieces into the reserved slots.
   \item If processing time has been inserted ahead of the dedicated setup of the slot, reschedule properly.
\end{enumerate}
\end{algorithm}

\section{Improvements of the running time}\label{sec:better_running_time}

In this section, we revisit the splittable and the setup time model.
For the former, we address the problem of the running time dependence in the number of machines $m$, and for both we present an improved rounding procedure, yielding a better running time.

\subsection{Splittable Model -- Machine Dependence}

In the splittable model, the number of machines $m$ may be super-polynomial in the input size, because it is not bounded by the number of jobs $n$.
Hence, we need to be careful already when defining the schedule in order to get a polynomially bounded output.
We say a machine is \textit{composite} if it contains more than one job, and we say it is \textit{plain} if it contains at most one job.
For a schedule with makespan $T$, we call each machine \textit{trivial} if it is plain and has load $T$ or if it is empty, and \textit{nontrivial} otherwise.
We say a schedule with makespan $T$ is \textit{simple}, if the number of nontrivial machines is bounded by $\binom{n}{2}$.

\begin{lemma}
If there is a schedule with makespan $T$ for $I$ there is also a simple schedule with makespan $T$.
\end{lemma}
\begin{proof}
Let there be a schedule $S$ with makespan $T$ for $I$. 
For the first step, let us assume there are more than $\binom{n}{2}$ composite machines.
In this case, there exist two machines $M_1$ and $M_2$ and two jobs $a,b \in \jobs, a\not = b$ such that both machines contain parts of both jobs since there are at most $\binom{n}{2}$ different pairs of jobs.
Let $t_{M_x}(y)$ be the processing time combined with the setup time of job $y \in \{a,b\}$ on machine $M_{x}$, $x \in \{1,2\}$.
W.l.o.g. let $t_{M_1}(a)$ be the smallest value of the four. 
We swap this job part and its setup time with some of the processing time of the job $b$ on machine $M_2$. 
If the processing time of $b$ on $M_2$ is smaller than $t_{M_1}(a)$, there is no processing time of $b$ on $M_2$ left and we can discard the setup time from $b$ on this machine.
We can repeat this step iteratively until there are at most $\binom{n}{2}$ machines containing more than one job.

In the second step, we shift processing time from the composite machines to the plain ones. 
We do this for each job until it is either not contained on a composite machine or each plain machine containing this job has load $T$. 
If the job is no longer contained on a composite machine, we shift the processing time of the job such that all except one machine containing this job has load $T$.
Since this job does not appear on any composite machines, their number can be bounded by $\binom{n-1}{2}$, by repeating the first step.
Therefore, the number of nontrivial machines is bounded by $\binom{n-i}{2} + i \leq \binom{n}{2}$ for some $i\in\set{0,\dots,n}$.
\end{proof}

For a simple schedule a polynomial representation of the solution is possible:
For each job, we give the number of trivial machines containing this jobs, or fix a first and last trivial machine belonging to this job.
This enables a polynomial encoding length of the output, given that the remaining parts of the jobs are not fragmented into too many parts, which can be guaranteed using the results of Section \ref{sec:EPTAS}. 

To guarantee that the MCIP finds a simple solution, we need to modify it a little. 
We have to ensure that nontrivial configurations are not used to often. 
We can do this by summing up the number of those configurations and bound them by $\binom{n}{2}$. 
Let $\confs' \subseteq \confs$ be the set of nontrivial configurations, i.e., the set of configurations containing more than one module or one module with size smaller than $T$.
We add the following globally uniform constraint to the MCIP:
\begin{align}
\sum_{C \in \confs'} x_C \leq \binom{|\bsjobs|}{2}
\end{align}
Since this is an inequality, we have to introduce a slack variable increasing the brick size by one.
Furthermore, the bound on the biggest number occurring in the input as well as the range of the variables has to be increased by a factor of $\Oh(n^2)$, yielding a slightly altered running time for the $\MCIP$ of:
\[2^{\Oh(\nicefrac{1}{\eps^4}\log\nicefrac{1}{\eps})}n^2\log^3 nm\]

The number of modules with maximum size denotes for each job in $\bsjobs$ how many trivial machines it uses. 
The other modules can be mapped to the nontrivial configurations and the jobs can be mapped to the modules. 

We still have to schedule the jobs in $\ssjobs$. 
We do this as described in the proof of Lemma \ref{lem:split_rounding1}. 
We fill the nontrivial machines greedily step by step starting with the jobs having the smallest processing time. 
When these machines are filled, there are some completely empty machines left.  
Now, we estimate how many machines can be completely filled with the current job $j$. 
This can be done, by dividing the remaining processing time by $T-s_i$ in $\mathcal{O}(1)$. 
The remaining part is scheduled on the next free machine. 
This machine is filled up with the next job and again the number of machines which can be filled completely with the rest of this new job is determined. 
These steps are iterated until all jobs in $\ssjobs$ are scheduled.
This greedy procedure needs at most $\mathcal{O}(|\bsjobs|(|\bsjobs|-1) + |\ssjobs|) = \mathcal{O}(n^2)$ operations. 
Therefore we can avoid the dependence in the number of machines and the overall running time is dominated by the time it takes to solve the $\MCIP$.

\subsection{Improved Rounding Procedures}

To improve the running time in the splittable and setup class model, we reduce the number of module sizes via a geometric and an arithmetic rounding step.
In both cases, the additional steps are performed following all the other simplification steps.
The basic idea is to include setup times together with their corresponding job pieces or batches of jobs respectively into containers with suitably rounded sizes and to model these containers using the modules.
The containers have to be bigger in size than the objects they contain and the load on a machine is given by the summed up sizes of the containers on the machine.
Let $H^*$ be a set of container sizes. 
Then a $H^*$-structured schedule is a schedule in which each setup time together with its corresponding job piece or batch of jobs is packed in a container with the smallest size $h\in H^*$ such that the summed up size of the setup time and the job piece or batch of jobs is upper bounded by $h$.

\subparagraph*{Splittable Model.}

Consider the instance $I_2$ for the splittable model described in Section~\ref{sec:EPTAS_slittable}.
In this instance, each setup and processing time is a multiple of $\eps^2 T$ and we are interested in a schedule of length $(1 + 2\eps)T$.
For each multiple $h$ of $\eps^2 T$, let $\tilde{h}=(1+\eps)^{\ceil{\log_{1+\eps}h/(\eps^2 T)}}\eps^2 T$ and $\bar{h} = \ceil{\tilde{h}/\eps^2 T}\eps^2 T$, and $\bar{H} = \sett{\bar{h}}{h\in\eps^2T\ZZ_{\geq 1},h\leq (1 + 2\eps)^2 T}$.
Note that $|\bar{H}|\in\Oh(1/\eps\log 1/\eps)$
\begin{lemma}
If there is a $((1 + 2\eps)T,L')$-schedule for $I_2$ in which the length of each job part is a multiple of $\eps^2 T$, there is also a $\bar{H}$-structured $((1 + 2\eps)^2T,L')$-schedule for $I_2$ with the same property.
\end{lemma}
\begin{proof}
Consider such a schedule for $I_2$ and a pair of setup time  $s$ and job piece $q$ scheduled on some machine.
Let $h=s+q$.
Stretching the schedule by $(1 + 2\eps)$ creates enough space to place the pair into a container of size $\bar{h}$, because $(1+\eps)h\leq \tilde{h}$, and $\eps h \leq \eps^2 T$, since $s\geq \eps T$.
\end{proof}

To implement this lemma into the procedure the processing time bounds $\bar{T}$ and $\breve{T}$ both have to be properly increased.
Modeling a $\bar{H}$-structured schedule can be done quite naturally:
We simply redefine the size $\size(M)$ of a module $M=(s,q)\in\mods$ to be $\bar{s+q}$.
With this definition, we have $|H|=|\bar{H}|=\Oh(1/\eps\log 1/\eps)$, yielding an improved running time for solving the $\MCIP$ of:
\[2^{\Oh(\nicefrac{1}{\eps^2}\log^3\nicefrac{1}{\eps})}n^2\log^2 m\log nm\]
Combining this with the results above and the considerations in Section \ref{sec:EPTAS_slittable} yields the running time claimed below Theorem \ref{thm:main_result}.

\subparagraph*{Setup Class Model.}

In the setup class model, an analogue approach also yields a reduced set of module sizes, that is, $|H|=\Oh(1/\eps\log 1/\eps)$.
Therefore, the $\MCIP$ can be solved in time:
\[2^{\Oh(\nicefrac{1}{\eps^{3}}\log^4\nicefrac{1}{\eps})}K^2\log n\log m \log Km\]
Hence, we get the running time claimed beneath Theorem \ref{thm:main_result}.

\section{Conclusion}

We presented a more advanced version of the classical configuration IP, showed that it can be solved efficiently using algorithms for $n$-fold IPs, and developed techniques to employ the new IP for the formulation of efficient polynomial time approximation schemes for three scheduling problems with setup times, for which no such algorithms were known before.

For further research the immediate questions are whether improved running times for the considered problems, in particular for the preemptive model, can be achieved; whether the MCIP can be solved more efficiently; and to which other problems it can be reasonably employed. 
From a broader perspective, it would be interesting to further study the potential of new algorithmic approaches in integer programming for approximation, and, on the other hand, further study the respective techniques themselfs.

\bibliography{references}

\end{document}